\documentclass[10pt,conference]{IEEEtran}
\usepackage[cmex10]{amsmath}
\usepackage{amssymb,amsfonts}
\usepackage{algorithmic}
\usepackage{graphicx}
\usepackage{textcomp}
\usepackage{xcolor}

\usepackage{url}

\usepackage[bookmarks=false,breaklinks]{hyperref}

\clearpage{}\RequirePackage{latexsym}
\RequirePackage{dsfont}

\mathcode`.="613A          
\mathcode`|="326A

\let\iff=\Leftrightarrow        
\let\implies=\Rightarrow

\let\sqleq=\sqsubseteq

\let\pleq=\sqleq          
\let\ple=\sqle            
          
\let\pgt=\sqgt

\def\land{\mathrel{\wedge}}

\def\preleq{\mathrel{
  \raise 2pt\hbox{$\mathop\sqsubset\limits_{\hbox{$\sim$}}$}
}}

\let\union=\cup

\let\join=\sqcup
\let\bigjoin=\bigsqcup
\let\meet=\sqcap
\newcommand{\down}{\mathord{\downarrow}}
\newcommand\dec[1]{{{\Downarrow}#1}}

\newcommand{\map}{\mathrel{\hookrightarrow}}

\newcommand\auxfun[1]{\expandafter\newcommand\csname #1\endcsname{
 \mathop{\hbox{$\mathsf{#1}$}}\nolimits}}
 \newcommand\af[1]{\mathop{\hbox{$\mathsf{#1}$}}\nolimits}

\let\max=\relax

\auxfun{max}
\auxfun{min}
\auxfun{dom}
\auxfun{ran}
\auxfun{fst}
\auxfun{snd}

\newcommand{\pow}[1]{{\cal P}(#1)}
\newcommand{\mpow}[1]{{\cal M}(#1)}

\newcommand{\nat}{\mathds{N}}

\newcommand{\ids}{\mathds{I}}

\newcommand{\quotient}[2]{#1 / #2}

\clearpage{}
\clearpage{}\newcommand\tup[1]{\langle #1 \rangle}

\newcommand\daf[1]{\af{#1}^\delta}

\newcommand\centerxy[1]{
  \begin{center}
    \begin{tabular}{c}
      \begin{xy}
        #1
      \end{xy}
    \end{tabular}
  \end{center}
}

\newcommand{\crdt}{\mathcal{L}}
\newcommand{\dm}{\m^\delta}

\auxfun{A}
\auxfun{B}
\auxfun{C}
\auxfun{D}
\auxfun{E}
\auxfun{BP}
\auxfun{RR}
\auxfun{m}
\auxfun{add}
\auxfun{receive}
\auxfun{send}
\auxfun{store}
\auxfun{operation}

\auxfun{inf}

\newcommand\algsinglecol[4]{
  \begin{algorithm}[#1]
    \SetInd{.7em}{0em}
    \SetKwIF{If}{ElseIf}{Else}{if}{}{else if}{else}{end}
    \SetKwFor{For}{for}{}{end}

    \DontPrintSemicolon
    \SetKwBlock{inputs}{inputs:}{}
    \SetKwBlock{state}{state:}{}
    \SetKwBlock{periodically}{periodically}{}
    \SetKwBlock{on}{on}{}
    \SetKwBlock{fun}{fun}{}
    \SetKwBlock{blank}{}{}
    \SetKw{return}{return}
    #2
    \caption{#3\vspace{-1em}}
    \label{#4}
  \end{algorithm}
}

\newcommand{\seta}{\{a\}}
\newcommand{\setb}{\{b\}}
\newcommand{\setc}{\{c\}}
\newcommand{\setab}{\{a, b\}}
\newcommand{\setac}{\{a, c\}}

\newcommand{\setbc}{\{b, c\}}
\newcommand{\setabc}{\{a, b, c\}}

\newcommand\counta[1]{\{\A_{#1}\}}
\newcommand\countb[1]{\{\B_{#1}\}}
\newcommand\countab[2]{\{\A_{#1}, \B_{#2}\}}

\newcommand{\rmark}{\ding{51}}
\newcommand{\wmark}{\ding{55}}

\newcommand\qline[1]{\textbf{line} \ref{#1}}

\def\BaseColor{gray!15}
\def\OptColor{gray!45}

\newcommand\HL[1]{\leavevmode\rlap{\hbox to \hsize{\color{#1}\leaders\hrule height .7\baselineskip depth .95ex\hfill}}}

\newcommand\HLBase[1]{\colorbox{\BaseColor}{#1}}
\newcommand\HLOpt[1]{\colorbox{\OptColor}{#1}}
\clearpage{}
\usepackage{nicefrac}

\usepackage{lipsum}

\usepackage{subcaption}

\usepackage[all]{xy}

\usepackage{amsthm}
\theoremstyle{plain}
\newtheorem{definition}{Definition}
\newtheorem{example}{Example}
\newtheorem{proposition}{Proposition}

\usepackage{pifont}

\usepackage[noline,noend,linesnumbered]{algorithm2e}
\usepackage{multicol}

\usepackage{cite}

\usepackage{booktabs}

\begin{document}
\bstctlcite{IEEEexample:BSTcontrol}

\title{Efficient Synchronization of State-based CRDTs
}

\author{\IEEEauthorblockN{Vitor Enes}
\IEEEauthorblockA{\textit{HASLab / INESC TEC} and \\
\textit{Universidade do Minho} \\
Portugal}
\and
\IEEEauthorblockN{Paulo S\'{e}rgio Almeida}
\IEEEauthorblockA{\textit{HASLab / INESC TEC} and \\
\textit{Universidade do Minho} \\
Portugal}
\and
\IEEEauthorblockN{Carlos Baquero}
\IEEEauthorblockA{\textit{HASLab / INESC TEC} and \\
\textit{Universidade do Minho} \\
Portugal}
\and
\IEEEauthorblockN{Jo\~{a}o Leit\~{a}o}
\IEEEauthorblockA{\textit{NOVA LINCS, FCT} and \\
\textit{Universidade NOVA de Lisbon} \\
Portugal}
}

\maketitle

\begin{abstract}
  To ensure high availability in large scale distributed systems,
\emph{Conflict-free Replicated Data Types} (CRDTs)
relax consistency by allowing immediate query and update
operations at the local replica, with no need for remote synchronization.
State-based CRDTs synchronize replicas by periodically sending
their full state to other replicas, which can become
extremely costly as the CRDT state grows.
Delta-based CRDTs address this problem by producing small incremental states (deltas) to be used in synchronization
instead of the full state.
However, current synchronization algorithms for
delta-based CRDTs induce redundant wasteful delta propagation,
performing worse than expected, and surprisingly, no better
than state-based.
In this paper we:
1) identify two sources of inefficiency in current synchronization algorithms for delta-based CRDTs;
2) bring the concept of join decomposition to state-based CRDTs;
3) exploit join decompositions to obtain optimal deltas and
4) improve the efficiency of synchronization algorithms; and finally,
5) experimentally evaluate the improved algorithms.
 \end{abstract}


\section{Introduction}

Large-scale distributed systems often resort to replication techniques to
achieve fault-tolerance and load distribution. These systems have to make a
choice between availability and low latency or strong consistency
\cite{abadi-cap,brewer-cap,gilbert-cap,golab-pacelc}, many times opting for the
first \cite{consistency1,consistency2}.  A common approach is to allow
replicas of some data type to temporarily diverge, making sure these replicas
will eventually converge to the same state in a deterministic way.
\emph{Conflict-free Replicated Data Types} (CRDTs) \cite{crdts1,crdts2} can be
used to achieve this. They are key components in modern geo-replicated systems,
such as Riak~\cite{riak}, Redis~\cite{redis}, and Microsoft Azure Cosmos
DB~\cite{cosmosdb}.

CRDTs come mainly in two flavors: \emph{operation-based} and
\emph{state-based}.  In both, queries and updates can be executed immediately
at each replica, which ensures availability (as it never needs to coordinate
beforehand with remote replicas to execute operations).  In operation-based
CRDTs \cite{pure-op,crdts1}, operations are disseminated assuming a reliable
dissemination layer that ensures exactly-once causal delivery of operations.

State-based CRDTs need fewer guarantees from the communication channel:
messages can be dropped, duplicated, and reordered.  When an update operation
occurs, the local state is updated through a mutator, and from time to time
(since we can disseminate the state at a lower rate than the rate of the
updates) the full (local) state is propagated to other replicas.

Although state-based CRDTs can be disseminated over unreliable communication
channels, as the state grows, sending the full state
becomes unacceptably costly. Delta-based CRDTs \cite{deltas1,deltas2} address
this issue by defining delta-mutators that return a delta ($\delta$),
typically much smaller than the full state of the replica, to be merged with
the local state.  The same $\delta$ is also added to an outbound
$\delta$-buffer, to be periodically propagated to remote replicas.
Delta-based CRDTs have been adopted in industry as part of Akka Distributed
Data
framework~\cite{akka-data} and IPFS~\cite{ipfs, ipfs-deltas}.

However, and somewhat unexpectedly, we have observed (Figure \ref{fig:problem})
that current
delta-propagation algorithms can still disseminate much
redundant state between replicas,
performing worse than envisioned, and no better than the state-based approach.
This anomaly becomes noticeable when concurrent update
operations always occur between synchronization rounds,
and it is partially justified due to inefficient redundancy detection in
delta-propagation.

\begin{figure}[t]
  \begin{center}
    \begin{minipage}{0.33\textwidth}
      \includegraphics[width=\textwidth,keepaspectratio]{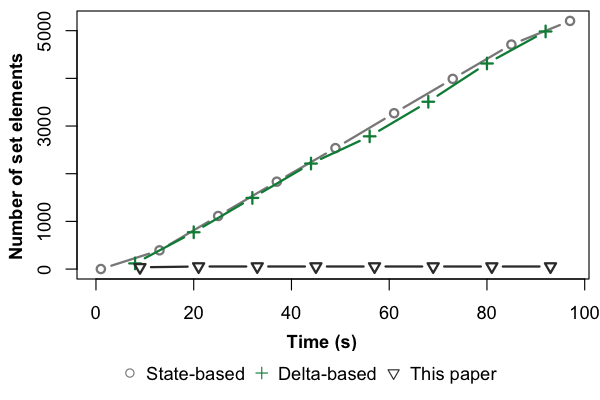}
    \end{minipage}
    \begin{minipage}{0.105\textwidth}
      \includegraphics[width=\textwidth,keepaspectratio]{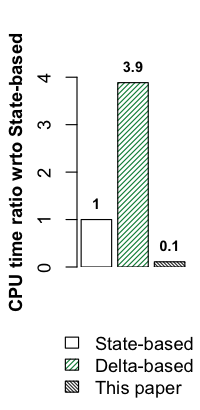}
    \end{minipage}
  \end{center}
  \caption{Experiment setup: 15 nodes in a partial mesh topology replicating
  an always-growing set. The left plot depicts the
  number of elements being sent throughout the experiment,
  while the right plot shows the CPU processing time ratio
  with respect to state-based. Not only does delta-based
  synchronization not improve state-based in terms of state transmission,
  it even incurs a substantial processing overhead.}
\label{fig:problem}
\end{figure}

In this paper we identify two sources of redundancy in current algorithms,
and introduce the concept of join decomposition of a state-based CRDT,
showing how it can be used to derive optimal deltas (``differences'') between states, as
well as optimal delta-mutators.
By exploiting these concepts, we also introduce an
improved synchronization algorithm, and experimentally
evaluate it, confirming that it outperforms current approaches
by reducing the amount of state transmission,
memory consumption, and processing time
required for delta-based synchronization.
 \section{Background on State-based CRDTs}

A state-based CRDT can be defined as a triple
$(\crdt, \sqleq, \join)$ where $\crdt$ is a
join-semilattice (lattice for short, from now on),
$\sqleq$ is a partial order, and
$\join$ is a binary join operator that derives
the least upper bound for any two elements of $\crdt$.
State-based CRDTs are updated through a set of
\emph{mutators} designed to be inflations,
i.e. for mutator $\m$ and state $x \in \crdt$, we have $x \sqleq \m(x)$.

Synchronization of replicas is achieved by having
each replica periodically propagate 
its local state to other
neighbour
replicas. When a remote state is received, a replica updates its state to reflect
the join of its local state and the received state.
As the local state grows, more state needs to be sent, which might affect the
usage of system resources (such as network) with a negative
impact on the overall system performance.
Ideally, each replica should only propagate the most
recent modifications executed over its local state.

Delta-based CRDTs can be used to achieve this,
by defining \emph{delta-mutators} that
return a smaller state which, when merged with the current state, generates the
same result as applying the standard mutators, i.e.
each mutator $\m$ has in
delta-CRDTs
a corresponding $\delta$-mutator
$\dm$ such that:
\[
  \m(x) = x \join \dm(x)
\]

In this model, the deltas resulting from $\delta$-mutators
are added to a $\delta$-buffer,
in order to be propagated to neighbor replicas, as a $\delta$-group, at the next synchronization step.
When a $\delta$-group is received from a neighbor, it is also added
to the buffer for further propagation.
 \subsection{CRDT examples}
\label{subsec:examples}

In Figure \ref{fig:crdt-spec} we present the specification
of two simple state-based CRDTs,
defining their lattice states, mutators, corresponding $\delta$-mutators,
and the binary join operator $\join$.
These lattices
are typically bounded and thus a bottom value
$\bot$ is also defined.
(Note that the specifications
do not define the partial order $\sqleq$ since it can
always be defined,
for any lattice $\crdt$, in terms of $\join$:
$x \sqleq y \iff x \join y = y$.)

\begin{figure}[!ht]
  \begin{center}
  \begin{subfigure}{0.48\textwidth}
    \begin{align*}
      \af{GCounter} & = \ids \map \nat \\
      \bot & = \varnothing \\
      \af{inc}_i(p) & = p\{i \mapsto p(i) + 1\} \\
      \daf{inc}_i(p) & = \{i \mapsto p(i) + 1\} \\
      \af{value}(p) & = \sum \{ v | k \mapsto v \in p \} \\
      p \join p' & = \{k \mapsto \max(p(k), p'(k)) | k \in l \} \\
      & \textbf{where } l = \dom(p) \union \dom(p')
    \end{align*}
    \caption{Grow-only Counter.}
    \label{fig:gcounter-spec}
  \end{subfigure}
  \end{center}
  \hfill
  \begin{center}
  \begin{subfigure}{0.48\textwidth}
    \begin{align*}
      \af{GSet}\tup{E} & = \pow{E} \\
      \bot & = \varnothing \\
      \af{add}(e, s) & = s \union \{e\} \\
      \daf{add}(e, s) & =
      \begin{cases}
        \{ e \} & \textbf{if}\enskip e \not \in s \\
        \bot & \textbf{otherwise}
      \end{cases} \\
      \af{value}(s) & = s \\
      s \join s' & = s \union s'
    \end{align*}
    \caption{Grow-only Set.}
    \label{fig:gset-spec}
  \end{subfigure}
  \end{center}
\caption{Specifications of two data types, replica $i \in \ids$.}
\label{fig:crdt-spec}
\end{figure}
 \begin{figure}[!ht]
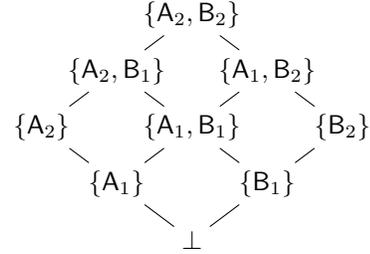
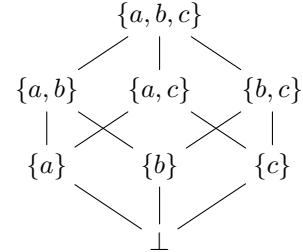

  \begin{center}
  \begin{subfigure}{0.48\textwidth}
    \centerxy{
			<0.5cm,0pt>:
      (0,6)*+{\countab{2}{2}}="atwobtwo";
      (-2,4.5)*+{\countab{2}{1}}="atwobone";
      (2,4.5)*+{\countab{1}{2}}="aonebtwo";
      (-4,3)*+{\counta{2}}="atwo";
      (0,3)*+{\countab{1}{1}}="aonebone";
      (4,3)*+{\countb{2}}="btwo";
      (-2,1.5)*+{\counta{1}}="aone";
      (2,1.5)*+{\countb{1}}="bone";
			(0,0)*+{\bot}="empty";
			"empty"; "aone" **\dir{-};
			"empty"; "bone" **\dir{-};
			"aone"; "atwo" **\dir{-};
			"bone"; "btwo" **\dir{-};
			"aone"; "aonebone" **\dir{-};
			"bone"; "aonebone" **\dir{-};
			"atwo"; "atwobone" **\dir{-};
			"aonebone"; "atwobone" **\dir{-};
			"btwo"; "aonebtwo" **\dir{-};
			"aonebone"; "aonebtwo" **\dir{-};
			"atwobone"; "atwobtwo" **\dir{-};
			"aonebtwo"; "atwobtwo" **\dir{-};
    }
    \caption{$\af{GCounter}$, with two replicas $\ids = \{\A, \B\}$.}
    \label{fig:gcounter-hasse}
  \end{subfigure}
  \end{center}
  \hfill
  \begin{center}
  \begin{subfigure}{0.48\textwidth}
    \centerxy{
      <0.5cm,0pt>:
      (0,6)*+{\setabc}="abc";
      (-3,4)*+{\setab}="ab";
      (0,4)*+{\setac}="ac";
      (3,4)*+{\setbc}="bc";
      (-3,2)*+{\seta}="a";
      (0,2)*+{\setb}="b";
      (3,2)*+{\setc}="c";
      (0,0)*+{\bot}="empty";
      "empty"; "a" **\dir{-};
      "empty"; "b" **\dir{-};
      "empty"; "c" **\dir{-};
      "a"; "ab" **\dir{-};
      "a"; "ac" **\dir{-};
      "b"; "ab" **\dir{-};
      "b"; "bc" **\dir{-};
      "c"; "ac" **\dir{-};
      "c"; "bc" **\dir{-};
      "ab"; "abc" **\dir{-};
      "ac"; "abc" **\dir{-};
      "bc"; "abc" **\dir{-};
    }
    \caption{$\af{GSet}\tup{\setabc}$.}
    \label{fig:gset-hasse}
  \end{subfigure}
  \end{center}
\caption{Hasse diagram of two data types.}
\label{fig:hasse}
\end{figure}

\newcommand{\setatob}{\{a, \underline b\}}
\newcommand{\setbtoa}{\{\underline a, \underline b, c\}}

\begin{figure*}[t]
  \centerxy{
\xymatrix@C=1.5em @R2em{
\A &
  \varnothing \ar@{->}[rr]^(0.45){\add_a} & & \seta
  \ar@{.}[rr] & & \setab
  \ar@{.}[r] & \bullet^2
  \ar@{->}[rd]^(.6){\setatob}
  \ar@{.}[rrr] & & & \setabc
\\
\B &
  \varnothing \ar@{->}[rr]^(0.45){\add_b} & & \setb
  \ar@{.}[r] & \bullet^1
  \ar@{->}[ru]_(.65){\setb}
  \ar@{.}[r] & \ar@{->}[r]^(0.45){\add_c} & \setbc
  \ar@{.}[r] & \setabc
  \ar@{.}[r] & \bullet^3
  \ar@{->}[ru]_(.65){\setbtoa}
}
}
  \caption{Delta-based synchronization of a $\af{GSet}$ with 2 replicas
  $\A, \B \in \ids$. Underlined elements represent the $\BP$ optimization.}
\label{fig:delta-ex1}
\end{figure*}
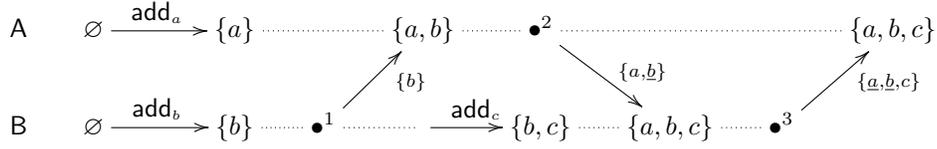
 
\newcommand{\setctod}{\{a, \overline b\}}

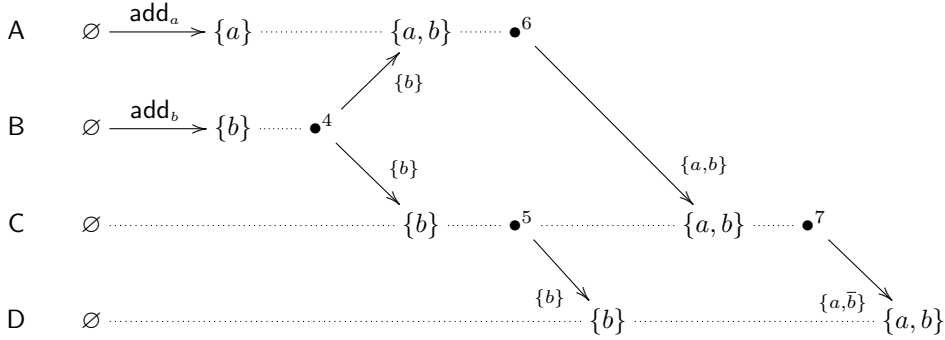
\begin{figure*}[t]
  \centerxy{
\xymatrix@C=1.5em @R2em{
\A &
\varnothing \ar@{->}[rr]^(0.45){\add_a} & & \seta
\ar@{.}[rr] & & \setab
\ar@{.}[r] & \bullet^6
\ar@{->}[rrdd]^(.78){\setab}
\\
\B &
\varnothing \ar@{->}[rr]^(0.45){\add_b} & & \setb
\ar@{.}[r] & \bullet^4
\ar@{->}[ru]_(.65){\setb} \ar@{->}[rd]^(.6){\setb}
\\
\C &
\varnothing
  \ar@{.}[rrrr] & & & & \setb
\ar@{.}[r] & \bullet^5
\ar@{->}[rd]_(.6){\setb}
  \ar@{.}[rr] & & \setab
\ar@{.}[r] & \bullet^7
\ar@{->}[rd]_(.6){\setctod}
\\
\D &
\varnothing
  \ar@{.}[rrrrrr] & & & & & & \setb
  \ar@{.}[rrr] & & & \setab
}
}
  \caption{Delta-based synchronization of a $\af{GSet}$ with 4 replicas
  $\A, \B, \C, \D \in \ids$. The overlined element represents the $\RR$ optimization.}
\label{fig:delta-ex2}
\end{figure*}
 
A CRDT counter that only allows increments is known as
a \emph{grow-only counter} (Figure \ref{fig:gcounter-spec}).
In this data type, the set of replica identifiers $\ids$
is mapped to the set of natural numbers $\nat$.
Increments are tracked per replica $i$, individually,
and stored in a map entry $p(i)$. The value of
the counter is the sum of each entry's value in the map.
Mutator $\af{inc}$ returns the updated map (the notation $p\{k \mapsto v\}$ indicates that only entry $k$ in the map $p$ is updated to a new value $v$, the remaining entries left unchanged), while
the $\delta$-mutator $\daf{inc}$ only returns the
updated entry. The join of two $\af{GCounter}$s computes,
for each key, the maximum of the associated values.

The lattice state evolution (either by mutation or
join of two states)
can also be understood by looking at
the corresponding Hasse diagram (Figure \ref{fig:hasse}).
For example,
state $\countab{1}{1}$ in Figure \ref{fig:gcounter-hasse}
(where $\A_1$ represents entry $\{\A \mapsto 1 \}$ in the map,
i.e. one increment registered by replica $\A$),
can result from an increment on $\counta{1}$ by $\B$,
from an increment on
$\countb{1}$
by $\A$, or from the join of these two states.

A \emph{grow-only set}, Figures \ref{fig:gset-spec}
and Figure \ref{fig:gset-hasse},
is a set data type that only allows element additions.
Mutator $\add$ returns the updated set, while $\daf{add}$
returns a singleton set with the added element
(in case it was not in the set already).
The join of two $\af{GSet}$s simply computes the set union.

Although we have chosen as running examples very simple CRDTs,
the results in this paper can be extended to more complex ones,
as we show in
Appendix \ref{app:compositions}.
For further coverage of delta-based CRDTs see \cite{deltas2}. 

 \subsection{Synchronization Cost Problem}

Figures \ref{fig:delta-ex1} and \ref{fig:delta-ex2} illustrate possible
distributed executions of the classic delta-based synchronization algorithm
\cite{deltas2}, with replicas of a \emph{grow-only-set}, all starting with a
bottom value $\bot = \varnothing$. (This classic algorithm is captured in
Algorithm \ref{alg:both}, covered in Section \ref{sec:revisited}.)
Synchronization with neighbors is represented by $\bullet$ and synchronization
arrows are labeled with the state sent, where we overline or underline
elements that are being redundantly sent and can be removed (thus improving
network bandwidth consumption) by employing two simple and novel optimizations
that we introduce next.

In Figure \ref{fig:delta-ex1}, we have two replicas $\A, \B \in \ids$
and each adds an element to the replicated set.
At $\bullet^1$, $\B$ propagates the content of the $\delta$-buffer,
i.e. $\setb$, to neighbour $\A$.
At $\bullet^2$, $\A$ sends to $B$ $\setab$,
i.e.
the join of $\seta$ from a local mutation,
and the received $\setb$ from $\B$,
even though $\setb$ came from $B$ itself.
By simply tracking the origin of each $\delta$-group in
the $\delta$-buffer, replicas can
\textbf{avoid back-propagation of $\delta$-groups} ($\BP$).

Before receiving $\setab$, $\B$ adds a new element $c$ to the set,
also adding $\setc$ to the $\delta$-buffer.
Upon receiving $\setab$, and since what was received produces changes in the local
state, $\B$ adds it to the $\delta$-buffer.
At $\bullet^3$, $\B$ propagates all new changes since last synchronization
with $\A$: $\setc$ from a local mutation, and $\setab$ from $\B$,
even though $\setab$ came from replica $\B$.
When $\A$ receives $\setabc$, it will also add it to the buffer to be further propagated.
Note that as long as this pattern keeps repeating (i.e. there's always a state change between
synchronizations), delta-based synchronization will propagate the same amount of state
as state-based synchronization would, representing no improvement.
This is illustrated in Figure \ref{fig:delta-ex1},
and demonstrated empirically in Section \ref{sec:eval}.

In Figure \ref{fig:delta-ex2}, we have four replicas $\A, \B, \C, \D \in \ids$,
and replicas $\A, \B$ add an element to the set.
At $\bullet^4$, $\B$ propagates the content of the $\delta$-buffer
to neighbours $\A$ and $\C$.
At $\bullet^5$, $\C$ propagates the received $\setb$ to $\D$.
At $\bullet^6$, $\A$ sends
the join of $\seta$ from a local mutation and the received $\setb$
to $\C$.
Upon receiving the $\delta$-group $\setab$, 
$\C$ adds it to the $\delta$-buffer
and sends it to $\D$ at $\bullet^7$.
However, part of this $\delta$-group has already been in the $\delta$-buffer
(namely $b$), and thus, has already been propagated.
This observation hints for another optimization:
\textbf{remove redundant state in received $\delta$-groups} ($\RR$),
before adding them to the $\delta$-buffer.

Both $\BP$ and $\RR$ optimizations are detailed in Section \ref{sec:revisited},
where we incorporate them into the delta-based synchronization algorithm
with few changes.
 \section{Join Decompositions and Optimal Deltas}
\label{sec:efficient}

In this section we introduce state decomposition in state-based CRDTs, by
exploiting the mathematical concept of \emph{irredundant join decompositions}
in lattices. We then demonstrate how this concept can be used
to derive deltas and delta-mutators that are optimal,
in the sense that they produce the
smallest delta-state possible.
In Section \ref{sec:revisited} we show how this
same concept plays a key role in the $\af{RR}$ optimization
briefly described in the previous section.
 \subsection{Join Decomposition of a State-based CRDT}
\label{subsec:jd}

\begin{definition}[Join-irreducible state]
  State $x \in \crdt$ is join-irreducible if it cannot
  result from the join of any finite set of states $F \subseteq \crdt$ not
  containing $x$:
\[
    x = \bigjoin F \implies x \in F
\]

\end{definition}

\begin{example}
Let the following $p_1$, $p_2$ and $p_3$ be $\af{GCounter}$ states,
and $s_1$, $s_2$ and $s_3$ be $\af{GSet}$ states.

\begin{center}
\begin{minipage}{.25\textwidth}
	\begin{itemize}
		\item[\rmark] $p_1 = \counta{5}$
		\item[\rmark] $p_2 = \countb{6}$
		\item[\wmark] $p_3 = \countab{5}{7}$
	\end{itemize}
\end{minipage}
\begin{minipage}{.23\textwidth}
	\begin{itemize}
		\item[\wmark] $s_1 = \bot$
		\item[\rmark] $s_2 = \seta$
		\item[\wmark] $s_3 = \setab$
	\end{itemize}
\end{minipage}
\end{center}

States $p_3$ and $s_3$ are not join-irreducible states,
since they can be decomposed into (i.e. result from the join of) two states
different from themselves: $\counta{5}$ and $\countb{7}$ for $p_3$,
$\seta$ and $\setb$ for $s_3$.
Bottom (e.g., $s_1$) is never join-irreducible, as it is the join over an empty
set $\bigjoin \varnothing$.

\end{example}

In a Hasse diagram of a finite lattice (e.g., in Figure \ref{fig:hasse})
the join-irreducibles are those elements with exactly one link below.
Given lattice $\crdt$, we use $\mathcal{J}(\crdt)$ for the set of all
join-irreducible elements of $\crdt$.

\begin{definition}[Join Decomposition]
Given a lattice state $x \in \crdt$, a set of join-irreducibles $D$
is a join decomposition \cite{birkhoff1937} of $x$
if its join produces $x$:
\[
  D \subseteq \mathcal{J}(\crdt) \land \bigjoin D = x
\]
\end{definition}

\begin{definition}[Irredundant Join Decomposition]
  A join decomposition D is irredundant if no element in it is redundant:
\[
  D' \subset D \implies \bigjoin D' \ple \bigjoin D
\]
\end{definition}

\begin{example}
\label{ex:jd}
Let $p = \countab{5}{7}$ be a $\af{GCounter}$ state,
$s = \setabc$ a $\af{GSet}$ state,
and consider the following sets of states as tentative decompositions of $p$
and $s$.

\begin{center}
\begin{minipage}{.25\textwidth}
	\begin{itemize}
		\item[\wmark] $P_1 = \{\counta{5}, \countb{6}\}$
    \item[\wmark] $P_2 = \{\counta{5}, \countb{6}, \countb{7}\}$
		\item[\wmark] $P_3 = \{\countab{5}{6}, \countb{7}\}$
		\item[\rmark] $P_4 = \{\counta{5}, \countb{7}\}$
	\end{itemize}
\end{minipage}
\begin{minipage}{.23\textwidth}
	\begin{itemize}
		\item[\wmark] $S_1 = \{\setb, \setc\}$
		\item[\wmark] $S_2 = \{\setab, \setb, \setc\}$
		\item[\wmark] $S_3 = \{\setab, \setc\}$
		\item[\rmark] $S_4 = \{\seta, \setb, \setc\}$
	\end{itemize}
\end{minipage}
\end{center}
  Only $P_4$ and $S_4$ are irredundant join decompositions of $p$ and $s$.
  $P_1$ and $S_1$ are not decompositions since their join does not result in $p$
  and $s$, respectively;
  $P_2$ and $S_2$ are decompositions but contain redundant elements,
  $\countb{6}$ and $\setb$, respectively;
  $P_3$ and $S_3$ do not have redundancy, but contain reducible elements
  ($S_2$ fails to be an irredundant join decomposition for the same reason,
  since its element $\setab$ is also reducible).
\end{example}

As we show in Appendix \ref{app:unique-jds} and \ref{app:compositions},
these irredundant decompositions exist, are unique,
and can be obtained for CRDTs used in practice.
Let $\dec{x}$ denote the unique decomposition of element $x$.
From the Birkhoff's Representation Theorem~\cite{latticesAndOrder},
decomposition $\dec{x}$ is given by the maximals of the
join-irreducibles below $x$:
\[
  \dec x = \max \{ r \in \mathcal J(\crdt) | r \pleq x \}
\]

As two examples,
given a $\af{GCounter}$ state $p$ and a $\af{GSet}$ state $s$,
their (quite trivial) irredundant decomposition is given by:
\[
  \hfill
  \dec{p} = \{ \{k \mapsto v\} | k \mapsto v \in p \}
  \qquad
  \dec{s} = \{ \{e\} | e \in s\}
  \hfill
\]

We argue that these techniques can be applied to most (practical) implementations of CRDTs 
found in industry. The interested reader can find generic decomposition rules in Appendix~\ref{app:decomposing}.
 \subsection{Optimal deltas and \texorpdfstring{$\delta$}{}-mutators}

Having a unique irredundant join decomposition, we can define a function which
gives the minimum delta, or ``difference'' in analogy to set difference,
between two states $a, b \in \crdt$:
\[
  \Delta(a, b) = \bigjoin \{ y \in \dec{a} | y \not \pleq b \}
\]
which when joined with $b$ gives $a \join b$, i.e. $\Delta(a, b) \join b = a \join b$.
It is minimum (and thus, optimal)
in the sense that it is smaller than any other $c$
which produces the same result:
$c \join b = a \join b \implies \Delta(a, b) \pleq c$.

If not carefully designed, $\delta$-mutators can be a source of redundancy
when the resulting $\delta$-state contains information that has already been incorporated
in the lattice state.
As an example, the original $\delta$-mutator $\daf{add}$ of $\af{GSet}$ presented in
\cite{deltas1} always returns a singleton set with the element to be added,
even if the element is already in the set
(in Figure \ref{fig:gset-spec} we have presented a definition of
$\daf{add}$ that is optimal). By resorting to function $\Delta$, minimum delta-mutators can be trivially derived
from a given mutator:
\[
  \dm(x) = \Delta(\m(x), x)
\]

 \section{Revisiting Delta-based Synchronization}
\label{sec:revisited}

Algorithm \ref{alg:both} formally describes
delta-based synchronization at replica $i$.
The algorithm contains lines that belong to
\HLBase{classic} delta-based synchronization \cite{deltas1,deltas2},
and lines with \HLOpt{$\BP$} and \HLOpt{$\RR$} optimizations,
while non-highlighted lines belong to both.
In classic delta-based synchronization,
each replica $i$ maintains a lattice state $x_i \in \crdt$
(\qline{alg:state}),
and a $\delta$-buffer $B_i \in \pow\crdt$ as a set of lattice states
(\qline{alg:buffer}).
When an update operation occurs (\qline{alg:delta-operation}),
the resulting $\delta$ is merged with the local state $x_i$ (\qline{alg:merge})
and added to the buffer (\qline{alg:add-buffer}),
resorting to function $\store$. Periodically,
the whole content of the $\delta$-buffer (\qline{alg:buffer-collect})
is propagated to neighbors (\qline{alg:buffer-send}).

For simplicity of presentation, we assume that communication channels between
replicas cannot drop messages (reordering and duplication is considered), and
that is why the buffer is cleared after each synchronization step
(\qline{alg:buffer-clear}). This assumption can be removed by simply tagging
each entry in the $\delta$-buffer with a unique sequence number, and by
exchanging acks between replicas: once an entry has been acknowledged by every
neighbour, it is removed from the $\delta$-buffer, as originally
proposed in \cite{deltas1}.

When a $\delta$-group is received
(\qline{alg:delta-receive}),
then it is checked whether it will induce an inflation in the local state
(\qline{alg:delta-check}). If this is the case, the $\delta$-group is merged with the local state and added to the buffer
(for further propagation),
resorting to the same function $\store$.
The precondition in \qline{alg:delta-check} appears to be harmless,
but it is in fact, the source of most redundant state propagated
in this synchronization algorithm. Detecting an inflation
is not enough, since almost always there's something new to
incorporate. Instead, synchronization algorithms must extract
from the received $\delta$-group the lattice state responsible for the
inflation, as done by the $\RR$ optimization.

Few changes are required in order to incorporate
this and the $\BP$ optimization in the classic
algorithm, as we show next. 
This happens because our approach encapsulates most of its complexity
in the computation of join decompositions and function $\Delta$.
The fact that few changes are required to the classic synchronization 
algorithm is a benefit, that will minimize the efforts in incorporating 
these techniques in existing implementations.

\algsinglecol{t}{
  \newcommand\tab{\hspace{0.7em}}
  \newcommand\block[3]{\leavevmode\rlap{\colorbox{#2}{\vphantom{$X_0^1$}\hspace{#1}}}\makebox[#1][l]{#3}}
  \newcommand\leftright[2]{\block{0.4\hsize}{\BaseColor}{#1}\block{0.57\hsize}{\OptColor}{ #2}}
  \SetKw{kwif}{if }
  \SetKw{kwfor}{for }
\inputs{
  $n_i \in \pow\ids$, set of neighbors \;
}
\vspace{0.2cm}
\state{
  $x_i \in \crdt$, $x_i^0 = \bot$ \; \label{alg:state}
  \leftright
  {$B_i \in \pow \crdt$, $B_i^0 = \varnothing$}
  {$B_i \in \pow{\crdt \times \ids}$, $B_i^0 = \varnothing$}
  \; \label{alg:buffer}
}
\vspace{0.2cm}
\on({$\operation_i(\dm)$}){
  \label{alg:delta-operation}
  $\delta = \dm(x_i)$ \;
  $\store(\delta, i)$ \;
}
\vspace{0.2cm}
\periodically(// synchronize){
  \label{alg:delta-ship}
  $\kwfor j \in n_i$ \;
    \leftright
    {\tab $d = \bigjoin B_i$}
    {$d = \bigjoin \{s | \tup{s, o} \in B_i \land o \not = j \}$}
    \; \label{alg:buffer-collect}
    \tab $\send_{i,j}(\af{delta}, d)$ \; \label{alg:buffer-send}
  $B_i' = \varnothing$ \; \label{alg:buffer-clear}
}
\vspace{0.2cm}
\on({$\receive_{j,i}(\af{delta}, d)$}){
  \label{alg:delta-receive}
  \leftright{}{$d = \Delta(d, x_i)$} \; \label{alg:delta-extract}
  \leftright{$\kwif d \not \sqleq x_i$}{$\kwif d \not = \bot$} \; \label{alg:delta-check}
    \tab $\store(d, j)$ \;
}
\vspace{0.2cm}
\fun({$\store(s, o)$}){
  $x_i' = x_i \join s$ \; \label{alg:merge}
  \leftright{$B_i' = B_i \union \{s\}$}{$B_i' = B_i \union \{\tup{s, o}\}$}
  \label{alg:add-buffer}
}
\vspace{0.2cm}
}
{Delta-based synchronization algorithms at replica $i \in \ids$:
\HLBase{classic} version
and version with \HLOpt{$\BP$} and \HLOpt{$\RR$} optimizations.}
{alg:both}
 
\paragraph*{Avoiding back-propagation of $\delta$-groups}
For $\BP$, each entry in the $\delta$-buffer is tagged
with its origin (\qline{alg:buffer} and \qline{alg:add-buffer}),
and at each synchronization step with neighbour $j$,
entries tagged with $j$ are filtered out
(\qline{alg:buffer-collect}).

\paragraph*{Removing redundant state in received $\delta$-groups}
A received $\delta$-group can contain redundant state,
i.e. state that has already been propagated to neighbors,
or state that is in the $\delta$-buffer $B_i$ still to be propagated.
This occurs in topologies where the underlying graph has cycles,
and thus,
nodes can receive the same information through different
paths in the graph.
In order to detect if a $\delta$-group has redundant state,
nodes do not need to keep everything in the $\delta$-buffer
or even inspect the $\delta$-buffer: it is enough to compare the
received $\delta$-group with the local lattice state $x_i$.
In classic delta-based synchronization,
received $\delta$-groups were added to
$\delta$-buffer only if they would strictly inflate
the local state (\qline{alg:delta-check}).
For $\RR$,
we extract from the $\delta$-group what strictly inflates
the local state $x_i$ (\qline{alg:delta-extract}),
and $\store$ it if it is different from bottom
(\qline{alg:delta-check}).
This extraction is achieved by selecting
which irreducible states
from the decomposition of the received $\delta$-group
strictly inflate the local state,
resorting to function $\Delta$
presented in Section \ref{sec:efficient}.

 \section{Evaluation}
\label{sec:eval}

In this Section we evaluate the proposed solutions and show the following:

\begin{itemize}
  \item Classic delta-based synchronization
    can be as inefficient as state-based synchronization in terms of
    transmission bandwidth, while incurring an overhead in terms of memory usage 
    required for synchronization (Section \ref{sub:micro}).
  \item In acyclic topologies, $\BP$ is enough to attain the best results,
    while in topologies with cycles, only $\RR$ can greatly reduce the synchronization
    cost (Section \ref{sub:micro}).
  \item Alternative synchronization techniques (such as Scuttlebutt \cite{scuttlebutt} and operation-based synchronization \cite{crdts1,crdts2}) are metadata-heavy; this metadata represents a large fraction of all the data required for synchronization (over 75\%)
    while for delta-based synchronization the metadata overhead can be as low as $7.7\%$
    (Section \ref{sub:micro}).
  \item In moderate-to-high contention workloads, $\BP$ + $\RR$
    can reduce transmission bandwidth and memory consumption by several GBs;
    when comparing with $\BP$ + $\RR$, classic delta-based synchronization has an unnecessary
    CPU overhead of up-to 7.9$\af{x}$ (Section \ref{sub:retwis}).
\end{itemize}

Instructions on how to reproduce all experiments can be found in
our public repository\footnote{\url{https://github.com/vitorenesduarte/exp}}.

\subsection{Experimental Setup}

The evaluation was conducted in a Kubernetes cluster deployed in Emulab \cite{emulab}.
Each machine has a Quad Core Intel Xeon 2.4 GHz and 12GB of RAM.
The number of machines in the cluster is set such that two replicas
are never scheduled to run in the same machine, i.e. there is at least
one machine available for each replica in the experiment.

\paragraph*{Network Topologies}
Figure \ref{fig:topologies} depicts the two network topologies
employed in the experiments:
a partial-mesh, in which each node has 4 neighbors;
and a tree, with 3 neighbors per node, with the exception of the root
node (2 neighbors) and leaf nodes (1 neighbor).
The first topology exhibits redundancy in the links and tests the effect of cycles in the
synchronization, while the second represents an optimal propagation scenario
over a spanning tree.

\begin{figure}[t]
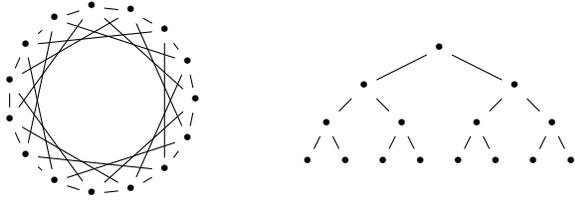

  \begin{minipage}{.24\textwidth}
  \centerxy{<0.5cm,0pt>:
  (2.5,0.0)*+{\vcenter{\hbox{\tiny$\bullet$}}}="0";
  (2.2838636441,1.0168416077)*+{\vcenter{\hbox{\tiny$\bullet$}}}="1";
  (1.6728265159,1.8578620637)*+{\vcenter{\hbox{\tiny$\bullet$}}}="2";
  (0.7725424859,2.3776412907)*+{\vcenter{\hbox{\tiny$\bullet$}}}="3";
  (-0.2613211582,2.4863047384)*+{\vcenter{\hbox{\tiny$\bullet$}}}="4";
  (-1.25,2.1650635095)*+{\vcenter{\hbox{\tiny$\bullet$}}}="5";
  (-2.0225424859,1.4694631307)*+{\vcenter{\hbox{\tiny$\bullet$}}}="6";
  (-2.4453690018,0.519779227)*+{\vcenter{\hbox{\tiny$\bullet$}}}="7";
  (-2.4453690018,-0.519779227)*+{\vcenter{\hbox{\tiny$\bullet$}}}="8";
  (-2.0225424859,-1.4694631307)*+{\vcenter{\hbox{\tiny$\bullet$}}}="9";
  (-1.25,-2.1650635095)*+{\vcenter{\hbox{\tiny$\bullet$}}}="10";
  (-0.2613211582,-2.4863047384)*+{\vcenter{\hbox{\tiny$\bullet$}}}="11";
  (0.7725424859,-2.3776412907)*+{\vcenter{\hbox{\tiny$\bullet$}}}="12";
  (1.6728265159,-1.8578620637)*+{\vcenter{\hbox{\tiny$\bullet$}}}="13";
  (2.2838636441,-1.0168416077)*+{\vcenter{\hbox{\tiny$\bullet$}}}="14";
  "5"; "9" **\dir{-};
  "10"; "11" **\dir{-};
  "4"; "8" **\dir{-};
  "5"; "6" **\dir{-};
  "0"; "14" **\dir{-};
  "8"; "9" **\dir{-};
  "3"; "7" **\dir{-};
  "7"; "11" **\dir{-};
  "3"; "14" **\dir{-};
  "1"; "2" **\dir{-};
  "6"; "7" **\dir{-};
  "12"; "13" **\dir{-};
  "6"; "10" **\dir{-};
  "0"; "11" **\dir{-};
  "1"; "5" **\dir{-};
  "2"; "13" **\dir{-};
  "0"; "4" **\dir{-};
  "2"; "6" **\dir{-};
  "4"; "5" **\dir{-};
  "9"; "10" **\dir{-};
  "2"; "3" **\dir{-};
  "11"; "12" **\dir{-};
  "0"; "1" **\dir{-};
  "9"; "13" **\dir{-};
  "1"; "12" **\dir{-};
  "8"; "12" **\dir{-};
  "13"; "14" **\dir{-};
  "3"; "4" **\dir{-};
  "10"; "14" **\dir{-};
  "7"; "8" **\dir{-};
}
   \end{minipage}
  \begin{minipage}{.24\textwidth}
  \centerxy{<0.5cm,0pt>:
  (-2.0,-1.0)*+{\vcenter{\hbox{\tiny$\bullet$}}}="(1, 2)";
  (0.0,0.0)*+{\vcenter{\hbox{\tiny$\bullet$}}}="(0, 1)";
  (2.5,-3.0)*+{\vcenter{\hbox{\tiny$\bullet$}}}="(3, 2)";
  (1.5,-3.0)*+{\vcenter{\hbox{\tiny$\bullet$}}}="(3, 3)";
  (3.5,-3.0)*+{\vcenter{\hbox{\tiny$\bullet$}}}="(3, 1)";
  (3.0,-2.0)*+{\vcenter{\hbox{\tiny$\bullet$}}}="(2, 1)";
  (2.0,-1.0)*+{\vcenter{\hbox{\tiny$\bullet$}}}="(1, 1)";
  (-3.5,-3.0)*+{\vcenter{\hbox{\tiny$\bullet$}}}="(3, 8)";
  (-1.0,-2.0)*+{\vcenter{\hbox{\tiny$\bullet$}}}="(2, 3)";
  (-1.5,-3.0)*+{\vcenter{\hbox{\tiny$\bullet$}}}="(3, 6)";
  (1.0,-2.0)*+{\vcenter{\hbox{\tiny$\bullet$}}}="(2, 2)";
  (-2.5,-3.0)*+{\vcenter{\hbox{\tiny$\bullet$}}}="(3, 7)";
  (0.5,-3.0)*+{\vcenter{\hbox{\tiny$\bullet$}}}="(3, 4)";
  (-3.0,-2.0)*+{\vcenter{\hbox{\tiny$\bullet$}}}="(2, 4)";
  (-0.5,-3.0)*+{\vcenter{\hbox{\tiny$\bullet$}}}="(3, 5)";
  "(2, 1)"; "(3, 1)" **\dir{-};
  "(0, 1)"; "(1, 2)" **\dir{-};
  "(2, 3)"; "(3, 6)" **\dir{-};
  "(1, 2)"; "(2, 3)" **\dir{-};
  "(2, 2)"; "(3, 4)" **\dir{-};
  "(0, 1)"; "(1, 1)" **\dir{-};
  "(2, 3)"; "(3, 5)" **\dir{-};
  "(2, 2)"; "(3, 3)" **\dir{-};
  "(1, 1)"; "(2, 2)" **\dir{-};
  "(1, 1)"; "(2, 1)" **\dir{-};
  "(2, 4)"; "(3, 7)" **\dir{-};
  "(2, 1)"; "(3, 2)" **\dir{-};
  "(1, 2)"; "(2, 4)" **\dir{-};
  "(2, 4)"; "(3, 8)" **\dir{-};
}
   \end{minipage}
  \caption{Network topologies employed:
  a 15-node partial-mesh (to the left)
  and a 15-node tree (to the right).}
  \label{fig:topologies}
\end{figure}
 
\subsection{Micro-Benchmarks}
\label{sub:micro}

We have designed a set of micro-benchmarks, in which
each node periodically
(every second)
synchronizes with neighbors and
executes an update operation over a CRDT. The update operation depends
on the CRDT type. In $\af{GSet}$, the update event is the addition of a globally
unique element to the set; 
in $\af{GCounter}$, an increment on the counter;
and in $\af{GMap\ K\%}$ each node updates $\frac{\af{K}}{\af{N}}\%$ keys
($\af{N}$ being the number of nodes/replicas),
such that globally $\af{K\%}$ of all the keys in the \emph{grow-only map}
are modified within each synchronization interval.
Note how the $\af{GCounter}$ benchmark is a particular case of $\af{GMap\ K\%}$,
in which $\af{K} = 100$.
For $\af{GMap\ K\%}$ we set the total number of keys to 1000,
and for all benchmarks, the number of events per replica is set to 100.

\newcommand\mltext[2]{
  \begin{minipage}{#1}
    \vspace*{.02cm}
    \begin{flushleft}
    #2
    \end{flushleft}
    \vspace*{-.42cm}
  \end{minipage}
}

\def\colalen{1.5cm}
\def\colblen{2.4cm}
\def\colclen{2.6cm}

\begin{table}[]
\caption{Description of micro-benchmarks.}
\label{tab:micro}
\begin{center}
  \begin{tabular}{c c p{\colclen}}
\toprule
    \mltext{\colalen}{\textbf{Type}}
      & \mltext{\colblen}{\textbf{Periodic event}}
      & \mltext{\colclen}{\textbf{Measurement}} \\ \toprule
    \mltext{\colalen}{$\af{GCounter}$}
      & \mltext{\colblen}{single increment}
      & \mltext{\colclen}{number of entries in the map} \\ \midrule
    \mltext{\colalen}{$\af{GSet}$}
      & \mltext{\colblen}{addition of unique element}
      & \mltext{\colclen}{number of elements in the set} \\ \midrule
    \mltext{\colalen}{$\af{GMap\ K\%}$}
      & \mltext{\colblen}{change the value of $\frac{\af{K}}{\af{N}}\%$ keys}
      & \mltext{\colclen}{number of entries in the map} \\ \bottomrule
\end{tabular}
\end{center}
\end{table}
 
These micro-benchmarks are summarized in Table \ref{tab:micro},
along with the metric (to be used in transmission and memory measurements)
we have defined:
for $\af{GCounter}$ and $\af{GMap\ K\%}$ we count the number of map entries,
while for $\af{GSet}$, the number of set elements.
We setup this part of the evaluation with 15-node topologies
(as in Figure \ref{fig:topologies}).
As baselines, we have state-based synchronization,
classic delta-based synchronization,
Scuttlebutt, a variation of Scuttlebutt,
and operation-based synchronization.

\paragraph*{Scuttlebutt}
Scuttlebutt \cite{scuttlebutt} is an anti-entropy protocol
used to reconcile changes in values of a key-value store. Each value is uniquely identified
with a version $\tup{i, s} \in \ids \times \nat$, where the first component $i \in \ids$
is the identifier of the replica responsible for the new value,
and $s \in \nat$ a sequence number, incremented on each local update, thus being unique.
With this, the updates known locally can be summarized by a vector $\ids \map \nat$,
mapping each replica to the highest sequence number it knows.
When a node wants to reconcile with a neighbor replica, it sends the summary vector,
and the neighbor replies with all the key-value pairs it has locally that have versions
not summarized in the received vector.
This strategy is performed in both directions, and in the end, both replicas
have the same key-value pairs in their local key-value store
(assuming no new updates occurred).

Scuttlebutt can be used to synchronize state-based CRDTs with few modifications.
Using as values the CRDT state would be inefficient, since changes to the CRDT
wouldn't be propagated incrementally, i.e. a small change in the CRDT would require
sending the whole new state, as in state-based synchronization.
Therefore, we use as values the optimal deltas resulting from $\delta$-mutators.
As keys, we can simply resort to the version pairs.
When reconciling two replicas, a replica receiving new key-delta pairs,
merges all the deltas with the local CRDT. If CRDT updates stop, eventually
all replicas converge to the same CRDT state. We label this approach \texttt{Scuttlebutt}.

\begin{figure}[t]
  \begin{center}
  \includegraphics[width=.46\textwidth,keepaspectratio]{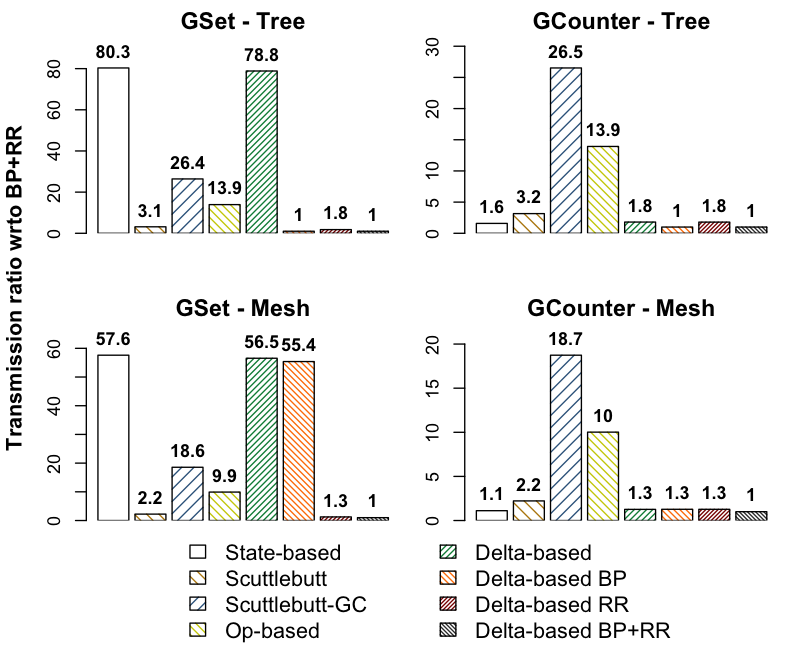}
  \end{center}
  \caption{Transmission of $\af{GSet}$ and $\af{GCounter}$
  with respect to delta-based $\BP + \RR$
  -- tree and mesh topologies.}
\label{fig:gset-gcounter-transmission}
\end{figure}

\begin{figure*}[t]
  \begin{center}
  \includegraphics[width=.92\textwidth,keepaspectratio]{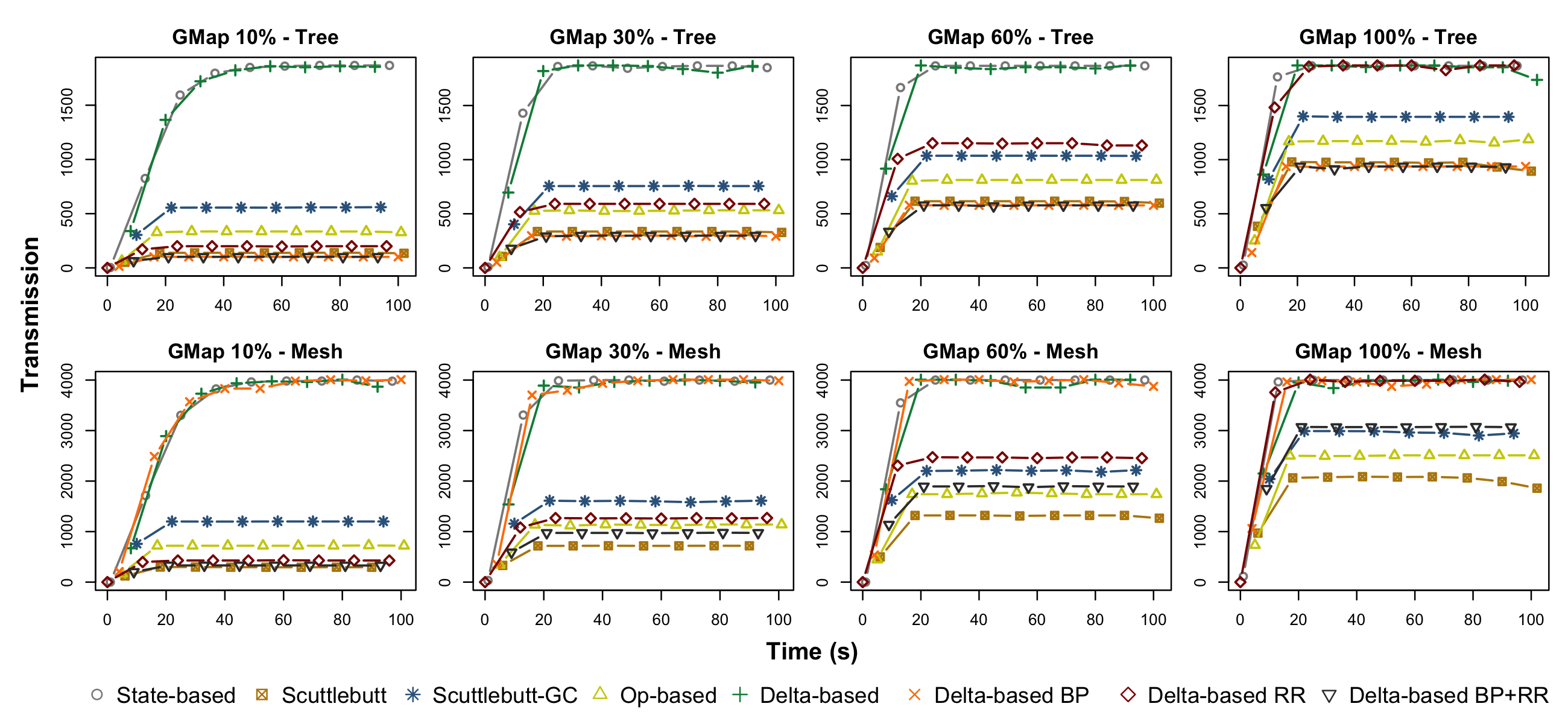}
  \end{center}
  \caption{Transmission of $\af{GMap\ 10\%}$,
  $\af{30\%}$, $\af{60\%}$ and $\af{100\%}$
  -- tree and mesh topologies.}
\label{fig:gmap-transmission}
\end{figure*}

This strategy is potentially inefficient in terms of memory: a replica has to keep in the
Scuttlebutt key-value store \emph{all} the deltas it has ever seen, since a neighbor
replica can at any point in time send a summary vector asking for \emph{any} delta.
Since the original Scuttlebutt algorithm does not support deleting keys from the key-value
store, we add support for \emph{safe} deletes of deltas,
in order to reduce its memory footprint.
If each node keeps track of what each node in the system has seen
(in a map $\ids \map (\ids \map \nat)$ from replica identifiers to
the last seen summary vector),
once a delta has been seen by all nodes,
it can be safely deleted from the local Scuttlebutt store.
We compare with this improved Scuttlebutt variant (labeled \texttt{Scuttlebutt-GC})
that allows nodes to only be connected
to a subset of all nodes, not requiring all-to-all connectivity,
while supporting safe deletes.
For completeness, we also compare with the original Scuttlebutt design
that is unable to garbage-collect unnecessary key-delta pairs.

\paragraph*{Operation-based}

Operation-based CRDTs \cite{crdts1,crdts2} resort to a causal broadcast
middleware~\cite{causal-multicast-survey} that is used to disseminate CRDT operations.
This middleware tags each operation with a vector clock that
summarizes the causal past of the operation.
Such vector is then used by the recipient to ensure causal delivery of operations,
i.e. each operation is only delivered when every operation in its causal past
has been delivered as well.

In topologies with all-to-all connectivity, each node is only responsible for disseminating
its own operations. In order to relax this requirement, we have implemented a middleware
that \emph{stores-and-forwards} operations: when an operation is seen for the first time, it is
added to a transmission buffer to be further propagated in the next synchronization step;
if the same operation is received from different incoming neighbors, the middleware simply
updates which nodes have seen this operation so that unnecessary transmissions are avoided.
To the best of our knowledge, this is the best possible implementation of such a middleware. 
We label this approach \texttt{Op-based}.

\subsubsection{Transmission bandwidth}

Figure \ref{fig:gset-gcounter-transmission}
shows, for $\af{GSet}$ and $\af{GCounter}$,
the transmission ratio (of all synchronization mechanisms previously mentioned)
with respect to delta-based synchronization with $\BP$ and $\RR$ optimizations enabled.
The first observation is that
classic delta-based synchronization presents almost
no improvement,
when compared to state-based synchronization.
In the tree topology, $\BP$ is enough to attain the best result,
because the underlying topology does not have cycles,
and thus, $\BP$ is sufficient to prevent redundant state to be propagated.
With a partial-mesh, $\BP$ has little effect,
and $\RR$ contributes most to the overall improvement.
Given that the underlying topology leads to redundant communication
(desired for fault-tolerance),
and classic delta-based can never extract that redundancy,
its transmission bandwidth is effectively similar to that of state-based synchronization.

Scuttlebutt and Scuttlebutt-GC are more efficient than classic delta-based for $\af{GSet}$
since both can precisely identify state changes between synchronization
rounds. However, the results for $\af{GCounter}$ reveal a limitation of this approach.
Since Scuttlebutt treats propagated values as opaque,
and does not understand that the changes in a $\af{GCounter}$ compress naturally
under lattice joins (only the highest sequence for each replica needs to be kept),
it effectively behaves worse than state-based and classic delta-based in this case.
Operation-based synchronization follows the \emph{same trend}
for the \emph{same reason}:
it improves state-based and classic delta-based for $\af{GSet}$
but not for $\af{GCounter}$ since the middleware is unable to compress
multiple operations into a single, equivalent, operation.
Supporting generic operation-compression at the middleware level
in operation-based CRDTs is an open research problem.
The difference between these three approaches is related with
the metadata cost associated to each,
as we show in Section~\ref{subsub:metadata}.

Even with the optimizations $\BP + \RR$ proposed,
the best result for $\af{GCounter}$ is not much
better than state-based.
This is expected since most entries of the underlying map are being updated
between each synchronization step:
each node has almost always something new from every other node
in the system to propagate
(thus being similar to state-based in some cases).
This pattern represents a special case of a map
in which $\af{100\%}$ of its keys are updated between
state synchronizations.

In Figure \ref{fig:gmap-transmission} we study
other update patterns, by measuring
the transmission of $\af{GMap\ 10\%}$, $\af{30\%}$, $\af{60\%}$,
and $\af{100\%}$.
These results are further evidence of what we have
observed
in the case of $\af{GSet}$:
$\BP$ suffices if the network graph is acyclic,
but $\RR$ is crucial in the more general case.

As seen previously, Scuttlebutt and Scuttlebutt-GC behave much better
than state-based synchronization, yielding a reduction in the transmission cost between
$46\%$ and $91\%$, and $20\%$ and $65\%$, respectively. This is
due to the underlying precise reconciliation mechanism of Scuttlebutt.
Operation-based synchronization leads to a transmission reduction between
$35\%$ and $80\%$ since it is able to represent
incremental changes to the CRDT as small operations.
Finally, delta-based $\BP + \RR$ is able reduce the transmission
costs by up-to $94\%$.

In the extreme case of $\af{GMap\ 100\%}$
(every key in the map is updated between synchronization rounds, which is a less likely workload in practical systems)
and considering a partial-mesh, delta-based $\BP + \RR$
provides a modest improvement in relation to state-based of about $18\%$
less transmission, and its performance is below Scuttlebutt variants and
operation-based synchronization.

Vector-based protocols (Scuttlebutt and operation-based) however, have an inherent scalability problem.
When increasing the number of nodes in the system, the transmission costs may
become
dominated by the size of metadata required for synchronization, as we show next.

\begin{figure}[t]
  \begin{center}
  \includegraphics[width=.32\textwidth,keepaspectratio]{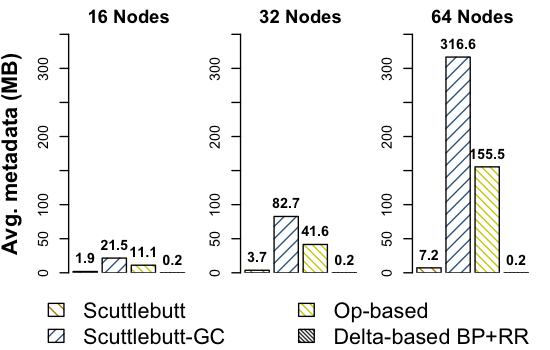}
  \end{center}
  \caption{Metadata required per node when synchronizing a $\af{GSet}$
  in a mesh topology. Each node has 4 neighbours (as in
  Figure \ref{fig:topologies}) and each node identifier has size 20B.}
  \label{fig:metadata}
\end{figure}

\begin{figure}[t]
  \begin{center}
  \includegraphics[width=.46\textwidth,keepaspectratio]{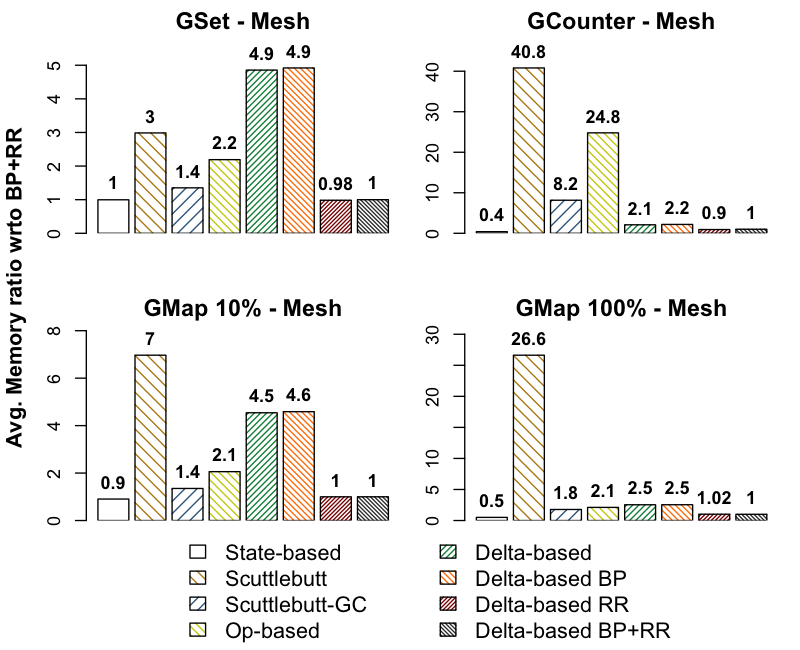}
  \end{center}
  \caption{Average memory ratio with respect to $\BP + \RR$
  for $\af{GCounter}$, $\af{GSet}$,
  $\af{GMap\ 10\%}$ and $\af{100\%}$
  -- mesh topology}
\label{fig:memory}
\end{figure}

\subsubsection{Metadata Cost}\label{subsub:metadata}
Figure~\ref{fig:metadata} shows the size of metadata required
for synchronization per node while varying the
total number of replicas (i.e. nodes). The results show a linear and quadratic cost 
(in terms of number of nodes) for Scuttlebutt and Scuttlebutt-GC (respectively),
and a 
linear cost for operation-based synchronization
(in terms of both number of nodes and pending updates still to be propagated).
Given $N$ nodes, $P$ neighbors, and $U$ pending updates,
the metadata cost per node is:
\begin{itemize}
  \item Scuttlebutt: $NP$ (a vector per neighbor)
  \item Scuttlebutt-GC: $N^2 P$ (a map of vectors per neighbor)
  \item Operation-based: $NPU$ (a vector per neighbor per pending update)
  \item Delta-based: $P$ (a sequence number per neighbor)
\end{itemize}

This cost may represent a large fraction of all data propagated during
synchronization.
For example, in our measurements with 32 nodes, this metadata
represents $75\%$, $99\%$, and $97\%$ of the transmission costs
for Scuttlebutt, Scuttlebutt-GC and operation-based, respectively,
while the overhead of delta-based synchronization is only $7.7\%$.

\subsubsection{Memory footprint}
In delta-based synchronization,
the size of $\delta$-groups being propagated not only affects
the network bandwidth consumption, but also
the memory required to store them in the $\delta$-buffer
for further propagation.
During the experiments, we periodically measure
the amount of state (both CRDT state and metadata required
for synchronization) stored in memory for each node.

Figure~\ref{fig:memory} reports the average
memory ratio with respect to $\BP + \RR$.
State-based does not require synchronization metadata,
and thus it is optimal in terms of memory usage.
Classic delta-based and delta-based $\BP$ have an overhead
of 1.1$\af{x}$-3.9$\af{x}$ since the size of $\delta$-groups in the $\delta$-buffer
is larger for these techniques.
For $\af{GSet}$ and $\af{GMap\ 10\%}$, Scuttlebutt-GC is close to $\BP + \RR$
since deltas are removed from the key-value store
as soon as they are seen by all replicas.
Key-delta pairs are never pruned in the original Scuttlebutt, leading to an increasing memory usage. 
As long as new updates exist, the memory consumption
for Scuttlebutt can only deteriorate, ultimately to a point where it will disrupt the system operation.
Operation-based has a higher memory cost than Scuttlebutt-GC, since each
operation in the transmission buffer is tagged with a vector,
while in Scuttlebutt and Scuttlebutt-GC each delta is simply tagged with a version pair.

Considering the results for $\af{GCounter}$, the three vector-based algorithms exhibit
the highest memory consumption. This is justified by 
the same reason they perform poorly
in terms of transmission bandwidth
in this case (Figure~\ref{fig:gset-gcounter-transmission}):
these protocols are unable to compress incremental changes.
Overall, and ignoring state-based which doesn't present any metadata memory costs,
$\BP + \RR$ attains the best results.

 \subsection{Retwis Application}
\label{sub:retwis}

\def\colalen{1.5cm}
\def\colblen{2cm}
\def\colclen{1.7cm}

\begin{table}[]
\caption{Retwis workload characterization: for each operation, the number of
  CRDT updates performed and its workload percentage.}
\label{tab:retwis}
\begin{center}
  \begin{tabular}{c c p{\colclen}}
\toprule
    \mltext{\colalen}{\textbf{Operation}}
      & \mltext{\colblen}{\textbf{\#Updates}}
      & \mltext{\colclen}{\textbf{Workload \%}} \\ \toprule
    \mltext{\colalen}{Follow}
      & \mltext{\colblen}{1}
      & \mltext{\colclen}{15\%} \\ \midrule
    \mltext{\colalen}{Post Tweet}
      & \mltext{\colblen}{1 + \#Followers}
      & \mltext{\colclen}{35\%} \\ \midrule
    \mltext{\colalen}{Timeline}
      & \mltext{\colblen}{0}
      & \mltext{\colclen}{50\%} \\ \bottomrule
\end{tabular}
\end{center}
\end{table}

We now compare classic delta-based with
delta-based $\BP + \RR$ using Retwis \cite{retwis},
a popular \cite{tapir,walter,tardis} open-source Twitter clone.
In Table \ref{tab:retwis} we describe the application workload,
similar to the one used in \cite{tapir}:
user $a$ can follow user $b$ by updating the set of followers of user $b$;
users can post a new tweet, by writing it in their wall and in the timeline
of all their followers;
and finally, users can read their timeline, fetching the 10 most recent tweets.

Each user has 3 objects associated with it:
1) a set of followers stored in a $\af{GSet}$;
2) a wall stored in a $\af{GMap}$ mapping tweet identifiers to their content; and
3) a timeline stored in a $\af{GMap}$ mapping tweet timestamps to tweet identifiers.
We run this benchmark with 10K users, and thus, 30K CRDT objects overall.
The size of tweet identifiers and content is 31B and 270B, respectively.
These sizes are representative of real workloads,
as shown in an analysis of Facebook's general-purpose key-value store
\cite{facebook-workload}. The topology is a partial-mesh, with 50 nodes,
each with 4 neighbors, as in Figure \ref{fig:topologies},
and updates on objects follow a Zipf distribution, with 
coefficients
ranging from 0.5 (low contention) to 1.5 (high contention)
\cite{tapir}.

\begin{figure}[t]
  \begin{center}
  \includegraphics[width=.48\textwidth,keepaspectratio]{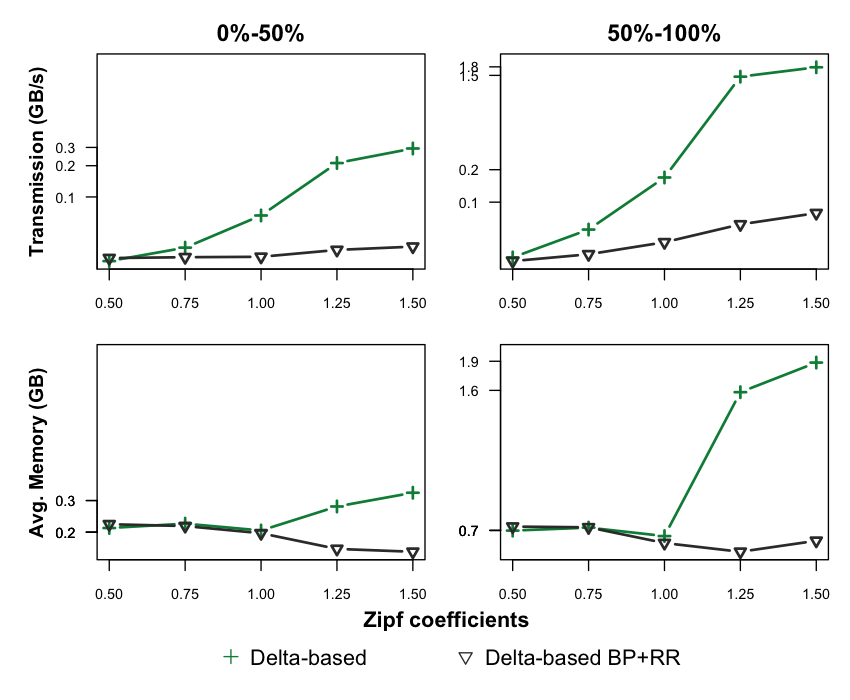}
  \end{center}
  \caption{Transmission bandwidth per node (top) and average memory per node (bottom) of
  classic delta-based and $\BP + \RR$ for different Zipf coefficient values
  (log scale).
  The left and right side show these values for
  the first and second half of the experiment (respectively).}
  \label{fig:retwis}
\end{figure}

Figure~\ref{fig:retwis}
shows the transmission bandwidth and memory footprint of both algorithms,
for different Zipf coefficient values.
We can observe that in low contention workloads, classic delta-based
behaves almost optimally when compared to $\BP + \RR$.
Since updates are distributed almost evenly across all objects,
there are few concurrent updates to the same object between synchronization rounds,
and thus, the simple and naive inflation check in~\qline{alg:delta-check}
suffices.
This phenomena was not observed in the previous set of benchmarks,
since we had a single object, and thus, maximum contention.

As we increase contention, a more sophisticated approach like $\BP + \RR$ is required,
in order to avoid redundant state propagation.
For example, with a 1.25 coefficient, bandwidth is reduced from $1.46$GB/s to $0.06$GB/s per node,
and memory footprint per node drops from $1.58$GB to $0.62$GB (right side of the plots).
Also, as we increase the Zipf coefficient, we note that the bandwidth consumption continues to 
rise, leading to an unsustainable situation in the case of classic delta-based,
as it can never reduce the size of $\delta$-groups being transmitted.

\begin{figure}[t]
  \begin{center}
  \includegraphics[width=.34\textwidth,keepaspectratio]{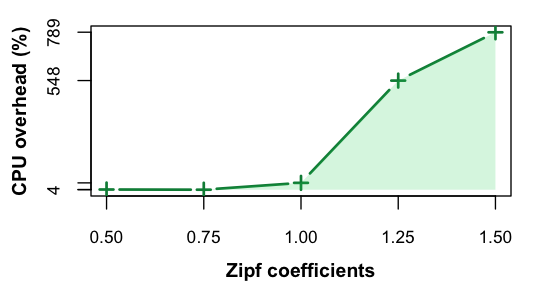}
  \end{center}
  \caption{CPU overhead of classic delta-based when compared to delta-based $\af{BP} + \af{RR}$.}
  \label{fig:retwis-processing}
\end{figure}

During the experiment we also measured the CPU time spent in processing CRDT
updates, both producing and processing synchronization messages.
Figure \ref{fig:retwis-processing} reports the CPU overhead of classic delta-based,
when considering $\BP + \RR$ as baseline.
Since classic delta-based produces/processes larger messages
than $\BP + \RR$, this results in a higher CPU cost:
for the 1, 1.25 and 1.5 Zipf coefficients, classic delta-based incurs an overhead
of 0.4$\af{x}$, 5.5$\af{x}$, and 7.9$\af{x}$ respectively.
  \section{Related Work}

In the context of remote file synchronization, \emph{rsync} \cite{rsync}
synchronizes two files placed on different machines,
by generating file block signatures,
and using these signatures to identify the missing blocks on the backup file.
In this strategy, there's a trade-off between the size of the blocks to be signed,
the number of signatures to be sent, and the size of the blocks to be received:
bigger blocks to be signed implies fewer signatures to be sent, but
the blocks received (deltas) can be bigger than necessary.
Inspired by \emph{rsync}, \emph{Xdelta}~\cite{xdelta} computes a difference between two files,
taking advantage of the fact that both files are present.
Consequently the cost of sending signatures can be ignored
and the produced deltas are optimized.

In \cite{join-decompositions}, we propose two techniques that can be used
to synchronize two state-based CRDTs after a network partition,
avoiding bidirectional full state transmission.
Let $\A$ and $\B$ be two replicas.
In \emph{state-driven} synchronization, $\A$
starts by sending its local lattice state to $\B$,
and given this state, $\B$ is able to compute a delta that reflects the updates missed by
$\A$.
In \emph{digest-driven} synchronization, $\A$ starts by sending a digest (signature)
of its local state (smaller than the local state), that still allows
$\B$ to compute the delta. $\B$ then sends the computed delta along with a digest
of its local state, allowing $\A$ to compute a delta for $\B$.
Convergence is achieved after 2 and 3 messages
in \emph{state-driven} and \emph{digest-driven}, respectively.
These two techniques also exploit the concept of join decomposition presented in this paper.

Similarly to \emph{digest-driven} synchronization,
$\Delta$-CRDTs \cite{making-deltas}
exchange metadata used to compute a delta that reflects missing updates.
In this approach, CRDTs need to be extended to maintain
additional metadata for delta derivation,
and if this metadata needs to be garbage collected,
the mechanism falls-back to standard bidirectional full state transmission.

In the context of anti-entropy gossip protocols,
\emph{Scuttlebutt} \cite{scuttlebutt}
proposes a \emph{push-pull} algorithm
to be used to synchronize a set of values between participants,
but considers each value as opaque, and does not try to represent
recent changes to these values as deltas. Other solutions try to minimize the communication overhead
of anti-entropy gossip-based protocols by exploiting either hash functions~\cite{clearhouse}
or a combination of Bloom filters, Merkle trees, and Patricia tries~\cite{Byers}. Still,
these solutions require a significant number of message exchanges to identify the source
of divergence between the state of two processes. Additionally, these solutions might incur 
significant processing overhead due to the need of computing hash functions and manipulating
complex data structures, such as Merkle trees.

With the exception of \emph{Xdelta}, all these techniques do not assume knowledge prior
to synchronization, and thus delay reconciliation,
by always exchanging state digests in order to detect state divergence.

 \section{Conclusion} \label{sec:conclusion}

Under geo-replication there is a significant availability and latency impact
\cite{abadi-cap} when aiming for strong consistency criteria such as linearizability
\cite{herlihy-linearizability}.
Strong consistency guarantees greatly simplify the programmers view of the system and are
still required for operations that do demand global synchronization.
However, several other system's components do not need that same level of coordination
and can reap the benefits of fast local operation and strong eventual consistency.
This requires capturing more information on each data type semantics, since a read/write
abstraction becomes limiting for the purpose of data reconciliation.
CRDTs can provide a sound approach to these highly available solutions and support the
existing industry solutions for geo-replication, which are still mostly grounded on state-based CRDTs. 

State-based CRDT solutions quickly become prohibitive in practice, if there is no support for
treatment of small incremental state deltas. In this paper we advance the foundations of
state-based CRDTs by introducing minimal deltas that precisely track state changes.
We also present and micro-benchmark two optimizations,
\emph{avoid back-propagation of $\delta$-groups} and
\emph{remove redundant state in received $\delta$-groups},
that solve inefficiencies in classic delta-based synchronization algorithms.
Further evaluation shows the improvement our solution can bring to
a small scale Twitter clone deployed in a 50-node cluster,
a relevant application scenario.

 \section*{Acknowledgments}
We would like to thank Ricardo Macedo, Georges Younes, Marc Shapiro and the anonymous reviewers for their valuable feedback on earlier drafts of this work.
Vitor Enes was supported by EU H2020 LightKone project (732505) and by a FCT - Funda{\c{c}}{\~{a}}o para a Ci{\^{e}}ncia e a Tecnologia - PhD Fellowship (PD/BD/142927/2018). Carlos Baquero was partially supported by SMILES within TEC4Growth project (NORTE-01-0145-FEDER-000020). Jo\~{a}o Leit\~{a}o was partially supported by project NG-STORAGE through FCT grant PTDC/CCI-INF/32038/2017,
and by NOVA LINCS through the FCT grant UID/CEC/04516/2013.
 
\bibliographystyle{IEEEtran}
\bibliography{bib}

\begin{thebibliography}{10}
\providecommand{\url}[1]{#1}
\csname url@samestyle\endcsname
\providecommand{\newblock}{\relax}
\providecommand{\bibinfo}[2]{#2}
\providecommand{\BIBentrySTDinterwordspacing}{\spaceskip=0pt\relax}
\providecommand{\BIBentryALTinterwordstretchfactor}{4}
\providecommand{\BIBentryALTinterwordspacing}{\spaceskip=\fontdimen2\font plus
\BIBentryALTinterwordstretchfactor\fontdimen3\font minus
  \fontdimen4\font\relax}
\providecommand{\BIBforeignlanguage}[2]{{%
\expandafter\ifx\csname l@#1\endcsname\relax
\typeout{** WARNING: IEEEtran.bst: No hyphenation pattern has been}%
\typeout{** loaded for the language `#1'. Using the pattern for}%
\typeout{** the default language instead.}%
\else
\language=\csname l@#1\endcsname
\fi
#2}}
\providecommand{\BIBdecl}{\relax}
\BIBdecl

\bibitem{abadi-cap}
D.~Abadi, ``{Consistency Tradeoffs in Modern Distributed Database System
  Design: CAP is Only Part of the Story},'' in \emph{Computer}, 2012.

\bibitem{brewer-cap}
E.~Brewer, ``{A Certain Freedom: Thoughts on the CAP Theorem},'' in
  \emph{PODC}, 2010.

\bibitem{gilbert-cap}
S.~Gilbert and N.~Lynch, ``{Brewer's Conjecture and the Feasibility of
  Consistent, Available, Partition-tolerant Web Services},'' in \emph{SIGACT
  News}, 2002.

\bibitem{golab-pacelc}
W.~Golab, ``{Proving PACELC},'' in \emph{SIGACT News}, 2018.

\bibitem{consistency1}
P.~Ajoux, N.~Bronson, S.~Kumar, W.~Lloyd, and K.~Veeraraghavan, ``{Challenges
  to Adopting Stronger Consistency at Scale},'' in \emph{HOTOS}, 2015.

\bibitem{consistency2}
H.~Lu, K.~Veeraraghavan, P.~Ajoux, J.~Hunt, Y.~J. Song, W.~Tobagus, S.~Kumar,
  and W.~Lloyd, ``{Existential Consistency: Measuring and Understanding
  Consistency at Facebook},'' in \emph{SOSP}, 2015.

\bibitem{crdts1}
M.~Shapiro, N.~M. Pregui{\c{c}}a, C.~Baquero, and M.~Zawirski, ``{Conflict-Free
  Replicated Data Types},'' in \emph{SSS}, 2011.

\bibitem{crdts2}
M.~Shapiro, N.~M. Pregui{\c{c}}a, C.~Baquero, and M.~Zawirski, ``{Convergent
  and Commutative Replicated Data Types},'' in \emph{Bulletin of the {EATCS}},
  2011.

\bibitem{riak}
\BIBentryALTinterwordspacing
Basho, ``{Riak KV Concepts: Data Types}.'' [Online]. Available:
  \url{http://docs.basho.com/riak/kv/2.2.3/learn/concepts/crdts/}
\BIBentrySTDinterwordspacing

\bibitem{redis}
\BIBentryALTinterwordspacing
R.~Labs, ``{Under the Hood: Redis CRDTs}.'' [Online]. Available:
  \url{https://redislabs.com/docs/active-active-whitepaper/}
\BIBentrySTDinterwordspacing

\bibitem{cosmosdb}
\BIBentryALTinterwordspacing
M.~Azure, ``{Multi-master at global scale with Azure Cosmos DB}.'' [Online].
  Available:
  \url{https://docs.microsoft.com/en-us/azure/cosmos-db/multi-region-writers}
\BIBentrySTDinterwordspacing

\bibitem{pure-op}
\BIBentryALTinterwordspacing
C.~Baquero, P.~S. Almeida, and A.~Shoker, ``{Pure Operation-Based Replicated
  Data Types},'' \emph{CoRR}, 2017. [Online]. Available:
  \url{http://arxiv.org/abs/1710.04469}
\BIBentrySTDinterwordspacing

\bibitem{deltas1}
P.~S. Almeida, A.~Shoker, and C.~Baquero, ``{Efficient State-Based CRDTs by
  Delta-Mutation},'' in \emph{NETYS}, 2015.

\bibitem{deltas2}
P.~S. Almeida, A.~Shoker, and C.~Baquero, ``{Delta State Replicated Data
  Types},'' in \emph{J. Parallel Distrib. Comput.}, 2018.

\bibitem{akka-data}
\BIBentryALTinterwordspacing
Akka, ``{Distributed Data}.'' [Online]. Available:
  \url{https://doc.akka.io/docs/akka/2.5/scala/distributed-data.html}
\BIBentrySTDinterwordspacing

\bibitem{ipfs}
\BIBentryALTinterwordspacing
IPFS, ``{Decentralized Real-Time Collaborative Documents}.'' [Online].
  Available: \url{https://ipfs.io/blog/30-js-ipfs-crdts.md}
\BIBentrySTDinterwordspacing

\bibitem{ipfs-deltas}
\BIBentryALTinterwordspacing
IPFS, ``{CRDT Research Repository}.'' [Online]. Available:
  \url{https://github.com/ipfs/research-CRDT/issues/31}
\BIBentrySTDinterwordspacing

\bibitem{birkhoff1937}
G.~Birkhoff, ``{Rings of sets},'' in \emph{Duke Mathematical Journal}, 1937.

\bibitem{latticesAndOrder}
B.~A. Davey and H.~A. Priestley, ``{Introduction to Lattices and
  Order}.''\hskip 1em plus 0.5em minus 0.4em\relax Cambridge University Press,
  1990.

\bibitem{scuttlebutt}
R.~van Renesse, D.~Dumitriu, V.~Gough, and C.~Thomas, ``{Efficient
  Reconciliation and Flow Control for Anti-entropy Protocols},'' in
  \emph{LADIS}, 2008.

\bibitem{emulab}
B.~White, J.~Lepreau, L.~Stoller, R.~Ricci, S.~Guruprasad, M.~Newbold,
  M.~Hibler, C.~Barb, and A.~Joglekar, ``{An Integrated Experimental
  Environment for Distributed Systems and Networks},'' in \emph{SIGOPS Oper.
  Syst. Rev.}, 2002.

\bibitem{causal-multicast-survey}
R.~Juan-Mar\'{\i}n, H.~Decker, J.~E. Armend\'{a}riz-\'{I}\~{n}igo, J.~M.
  Bernab{\'e}u-Aub\'{a}n, and F.~D. Mu\~{n}oz Esco\'{\i}, ``{Scalability
  Approaches for Causal Multicast: A Survey},'' in \emph{Distributed
  Computing}, 2016.

\bibitem{retwis}
\BIBentryALTinterwordspacing
Retwis. [Online]. Available: \url{http://retwis.antirez.com}
\BIBentrySTDinterwordspacing

\bibitem{tapir}
I.~Zhang, N.~K. Sharma, A.~Szekeres, A.~Krishnamurthy, and D.~R.~K. Ports,
  ``{Building Consistent Transactions with Inconsistent Replication},'' in
  \emph{SOSP}, 2015.

\bibitem{walter}
Y.~Sovran, R.~Power, M.~K. Aguilera, and J.~Li, ``{Transactional storage for
  geo-replicated systems},'' in \emph{SOSP}, 2011.

\bibitem{tardis}
N.~Crooks, Y.~Pu, N.~Estrada, T.~Gupta, L.~Alvisi, and A.~Clement, ``{TARDiS:
  {A} Branch-and-Merge Approach To Weak Consistency},'' in \emph{SIGMOD}, 2016.

\bibitem{facebook-workload}
B.~Atikoglu, Y.~Xu, E.~Frachtenberg, S.~Jiang, and M.~Paleczny, ``{Workload
  Analysis of a Large-Scale Key-Value Store},'' in \emph{SIGMETRICS}, 2012.

\bibitem{rsync}
A.~Tridgell and P.~Mackerras, ``{The rsync algorithm},'' Australian National
  University, Tech. Rep., 1998.

\bibitem{xdelta}
\BIBentryALTinterwordspacing
J.~Macdonald, ``{Xdelta}.'' [Online]. Available: \url{http://xdelta.org}
\BIBentrySTDinterwordspacing

\bibitem{join-decompositions}
V.~Enes, C.~Baquero, P.~S. Almeida, and A.~Shoker, ``{Join Decompositions for
  Efficient Synchronization of CRDTs after a Network Partition: Work in
  progress report},'' in \emph{PMLDC@ECOOP}, 2016.

\bibitem{making-deltas}
A.~van~der Linde, J.~Leit{\~a}o, and N.~Pregui{\c c}a, ``{{$\Delta$}-CRDTs:
  Making {$\Delta$}-CRDTs Delta-based},'' in \emph{PaPoC@EuroSys}, 2016.

\bibitem{clearhouse}
A.~Demers, D.~Greene, C.~Hauser, W.~Irish, J.~Larson, S.~Shenker, H.~Sturgis,
  D.~Swinehart, and D.~Terry, ``{Epidemic Algorithms for Replicated Database
  Maintenance},'' in \emph{PODC}, 1987.

\bibitem{Byers}
J.~Byers, J.~Considine, and M.~Mitzenmacher, ``{Fast Approximate Reconciliation
  of Set Differences},'' CS Dept., Boston University, Tech. Rep., 2002.

\bibitem{herlihy-linearizability}
M.~P. Herlihy and J.~M. Wing, ``{Linearizability: A Correctness Condition for
  Concurrent Objects},'' in \emph{Trans. Program. Lang. Syst.}, 1990.

\bibitem{crdt-composition}
C.~Baquero, P.~S. Almeida, A.~Cunha, and C.~Ferreira, ``{Composition in
  State-based Replicated Data Types},'' in \emph{Bulletin of the {EATCS}},
  2017.

\bibitem{single-writer}
V.~Enes, P.~S. Almeida, and C.~Baquero, ``{The Single-Writer Principle in CRDT
  Composition},'' in \emph{PMLDC@ECOOP}, 2017.

\bibitem{cassandra-counters}
\BIBentryALTinterwordspacing
DataStax, ``{What's New in Cassandra 2.1: Better Implementation of Counters}.''
  [Online]. Available:
  \url{https://www.datastax.com/dev/blog/whats-new-in-cassandra-2-1-a-better-implementation-of-counters}
\BIBentrySTDinterwordspacing

\end{thebibliography}

\appendix

\subsection{Existence of Unique Irredundant Decompositions}
\label{app:unique-jds}

In this section we present sufficient conditions for
the existence of unique irredundant join decompositions,
and show how they can be obtained.

\begin{definition}[Descending chain condition]
  A lattice $\crdt$ satisfies the \emph{descending chain condition} (DCC)
  if any sequence $x_1 \pgt x_2 \pgt \cdots \pgt x_n \pgt \cdots$
  of elements in $\crdt$ has finite length~\cite{latticesAndOrder}.
\end{definition}

\begin{proposition}
\label{prop:unique}
In a distributive lattice $\crdt$ satisfying DCC every element $x \in \crdt$ has a unique
irredundant join decomposition.
\end{proposition}
\begin{proof}
Trivial, as corollary of the dual of Theorem 6 from~\cite{birkhoff1937}: a distributive
lattice is modular; if it also satisfies DCC, then each element has a unique
irredundant join decomposition.
\end{proof}

For almost all CRDTs used in practice, the state is
not merely a join-semilattice, but a distributive lattice satisfying DCC (Appendix \ref{app:compositions}).
Therefore, from Proposition \ref{prop:unique},
we have a unique irredundant join decomposition for each CRDT state.
Let $\dec{x}$ denote this unique decomposition of an element $x$.

\begin{proposition}
  \label{prop:dec}
  If $\crdt$ is a finite distributive lattice, then $\dec{x}$ is given by the
  maximals of the join-irreducibles below $x$:
\[
  \dec x = \max \{ r \in \mathcal J(\crdt) | r \pleq x \}
\]
\end{proposition}
\begin{proof}
From the Birkhoff's Representation Theorem (see, e.g.,~\cite{latticesAndOrder}),
each element $x$ is isomorphic to $\{ r \in \mathcal J(\crdt) | r \pleq x \}$,
  the set of join-irreducibles below it, which is isomorphic to the set
  of its maximals, containing no redundant element.
\end{proof}

\begin{table*}[t]
\caption{Composition techniques that yield lattices satisfying DCC and
  distributive lattices, given lattices $A$ and $B$, chain $C$, partial order
  $P$ and (unordered) set $U$.}
\label{tab:dcc-dist}

\begin{center}
\begin{tabular}{c|c|c|c|c|c|c|c|}
\cline{2-8}
\multicolumn{1}{l|}{} & \multicolumn{7}{c|}{$\crdt$} \\ \cline{2-8} 
  & $A \times B$ & $A \boxtimes B$ & $C \boxtimes A$ & $A \oplus B$ & $U {\hookrightarrow} A$ & $\pow{U}$ & $\mpow{P}$ \\ \hline
\multicolumn{1}{|c|}{$A, B, P \textrm{ has DCC} \implies \crdt \textrm{ has DCC}$}
  & \rmark & \rmark & \rmark & \rmark & \rmark& \rmark  & \rmark \\ \hline
\multicolumn{1}{|c|}{$A, B \textrm{ distributive} \implies \crdt \textrm{ distributive}$}
  & \rmark & \wmark & \rmark & \rmark & \rmark& \rmark  & \rmark \\ \hline
\end{tabular}
\end{center}

\end{table*}
 \begin{table*}[t]
\caption{Composition techniques that yield finite ideals or quotients, given
  lattices $A$ and $B$, chain $C$, partial order $P$, all satisfying DCC, and (unordered) set $U$.}
\label{tab:finite}

\begin{center}
\begin{tabular}{c|c|c|c|c|c|c|c|}
\cline{2-8}
\multicolumn{1}{l|}{} & \multicolumn{7}{c|}{$\crdt$} \\ \cline{2-8} 
  & $A \times B$ & $A \boxtimes B$ & $C \boxtimes A$ & $A \oplus B$ & $U {\hookrightarrow} A$ & $\pow{U}$ & $\mpow{P}$ \\ \hline
\multicolumn{1}{|c|}{$\forall x \in \crdt \cdot x / \bot \textrm{ finite}$}
  & \rmark & \wmark & \wmark & \wmark & \rmark & \rmark  & \rmark \\ \hline
\multicolumn{1}{|c|}{$\forall \tup{x, y} \in \crdt \cdot \tup{x, y} / \tup{x, \bot} \textrm{ finite}$}
  & --     & \rmark & \rmark & \rmark     & --     & --      & --     \\ \hline
\end{tabular}
\end{center}

\end{table*}
 
Although Proposition \ref{prop:dec} is stated for finite lattices, 
it can be applied to typical CRDTs defined over infinite lattices,
as we show next.
 \subsection{Lattice Compositions in CRDTs}
\label{app:compositions}

We now show that unique irredundant join decompositions (and therefore,
optimal deltas and delta-mutators) can be obtained for almost all state-based
CRDTs used in practice. Most CRDT designs define the lattice state starting
from lattice chains (booleans and natural numbers), unordered sets,
partial orders, and obtain more complex states by lattice composition
through: cartesian product $\times$, lexicographic product $\boxtimes$, linear
sum $\oplus$, finite functions ${\hookrightarrow}$ from a set to a lattice, powersets
$\mathcal{P}$, and sets of maximal elements $\mathcal{M}$ (in a partial
order).
Note that two of the constructs, ${\hookrightarrow}$ and $\mathcal{P}$,
were used in Section \ref{subsec:examples}
to define $\af{GCounter}$ and $\af{GSet}$, respectively.
The use of these composition techniques and a catalog of CRDTs is presented
in~\cite{crdt-composition} but that presentation (as well as CRDT designs in
general) simply considers building join-semilattices (typically with bottom)
from join-semilattices, never examining whether the result is more than a
join-semilattice.

\newcommand\maximalsbelow{
  \{
    \tup{\seta, \seta},
    \tup{\seta, \setb},
    \tup{\setb, \seta},
    \tup{\setb, \setb}
  \}
}

In fact, all those constructs yield lattices with bottom when starting
from lattices with bottom. Moreover, all these constructs yield
lattices satisfying DCC, when starting from lattices satisfying DCC (such as
booleans and naturals). Also, it is easily seen that most yield distributive
lattices when applied to distributive lattices, with the exception of
the lexicographic product with an arbitrary first component.
As an example,
in Figure~\ref{fig:lex-dist} we depict the Hasse diagram of a non-distributive
lexicographic pair.
This lattice is non-distributive since, e.g.,
for $x = \tup{\seta, \seta}$, $y = \tup{\seta, \varnothing}$
and $z = \tup{\setb, \varnothing}$,
we have $x = x \meet (y \join z) \neq (x \meet y) \join (x \meet z) = y$.
For the join-reducible $\tup{\setab, \varnothing}$, the set of the maximals of
the join-irreducibles below it (i.e. $\maximalsbelow$) is a redundant decomposition (as well as some
of its subsets), and there are several alternative irredundant decompositions:

\begin{figure}[t]
  \centerxy{
    <0.5cm,0pt>:
    (0,7)*+{}="top";
    (0,6)*+{\tup{ \setab, \varnothing }}="abab";
    (-4,4.5)*+{\tup{ \seta, \setab }}="aab";
    (4,4.5)*+{\tup{ \setb, \setab }}="bab";
    (-6,3)*+{\tup{ \seta, \seta }}="aa";
    (-2,3)*+{\tup{ \seta, \setb }}="ab";
    (2,3)*+{\tup{ \setb, \seta }}="ba";
    (6,3)*+{\tup{ \setb, \setb }}="bb";
    (-4,1.5)*+{\tup{ \seta, \varnothing }}="a0";
    (4,1.5)*+{\tup{ \setb, \varnothing }}="b0";
    (0,0)*+{\tup{ \varnothing, \varnothing }}="00";
    "00"; "a0" **\dir{-};
    "00"; "b0" **\dir{-};
    "a0"; "aa" **\dir{-};
    "a0"; "ab" **\dir{-};
    "b0"; "ba" **\dir{-};
    "b0"; "bb" **\dir{-};
    "aa"; "aab" **\dir{-};
    "ab"; "aab" **\dir{-};
    "ba"; "bab" **\dir{-};
    "bb"; "bab" **\dir{-};
    "aab"; "abab" **\dir{-};
    "bab"; "abab" **\dir{-};
    "abab"; "top" **\dir{.};
  }
  \caption{Hasse diagram of $\pow\setab \boxtimes \pow\setab$, a non-distributive lattice.}
  \label{fig:lex-dist}
\end{figure}
 
\begin{center}
  \vspace{5pt}
\begin{minipage}{.47\linewidth}
\begin{itemize}
  \item $\{\tup{\seta, \varnothing}, \tup{\setb, \varnothing}\}$
  \item $\{\tup{\seta, \varnothing}, \tup{\setb, \seta}\}$
  \item $\{\tup{\seta, \varnothing}, \tup{\setb, \setb}\}$
\end{itemize}
\end{minipage}
\begin{minipage}{.47\linewidth}
\begin{itemize}
  \item $\{\tup{\seta, \seta}, \tup{\setb, \varnothing}\}$
  \item $\dots$
  \item $\{\tup{\seta, \setb}, \tup{\setb, \setb}\}$
\end{itemize}
\end{minipage}
  \vspace{4pt}
\end{center}

Fortunately,
the typical use of lexicographic products to design CRDTs is with a chain
(total order) as the first component, to allow an actor which is ``owner'' of
part of the state (the single-writer principle~\cite{single-writer})
to either inflate the second component, or to change it to some arbitrary value,
while increasing a ``version number'' (first component).
This principle is followed by Cassandra counters \cite{cassandra-counters}.
In such typical usages of the lexicographic product, with a chain as first component, the
distributivity of the second component is propagated to the resulting construct.
Table \ref{tab:dcc-dist} summarizes these remarks about how almost always
these CRDT composition techniques yield lattices satisfying DCC and distributive
lattices,
and thus, have unique irredundant decompositions, by Proposition \ref{prop:unique}.

Having DCC and distributivity, even if it always occurs in practice, is not
enough to directly apply Proposition~\ref{prop:dec}, as it holds for finite lattices.
However if the sublattice given by the ideal $\down x = \{ y |  y \pleq x \}$
is finite, then we can apply that proposition to this finite lattice (for
which $x$ is now the top element) to compute $\dec{x}$. Again, finiteness yields
from all constructs, with the exception of the lexicographic product and
linear sum. For these two constructs, a similar reasoning can be applied, but
focusing on a quotient sublattice in order to achieve finiteness.

\begin{definition}[Quotient sublattice]
  \label{def:quotient}
  Given elements $a \pleq b \in \crdt$, the \emph{quotient sublattice} $\quotient{b}{a}$
  is given by:
  \[
    \quotient{b}{a} = \{ x \in \crdt | a \sqleq x \sqleq b \}
  \]
\end{definition}

Quotients generalize ideals, as we have $\down x = \quotient{x}{\bot}$. As an
example, given some infinite set $U$ and the lattice $\nat \boxtimes \pow{U}$,
for each $x = \tup{n, s}$, the ideal $\down x$ is still infinite when $n > 0$,
as depicted in Figure~\ref{fig:lex-finite}.  However, for each $\tup{n, s}$,
the quotient $\quotient{\tup{n, s}}{\tup{n, \bot}}$ is a finite lattice, and
moreover, the elements given by $\dec{\tup{n, s}}$ are the same either when
considering the original lattice or the quotient sublattice. Therefore, we can
use the formula for $\dec{x}$ in Proposition~\ref{prop:dec}. A similar
reasoning can be used for linear sums.
Table \ref{tab:finite} summarizes these remarks; the second row applies only to
lexicographic products and linear sums\footnote{In order to
have a common notation for instances of $\boxtimes$ and $\oplus$,
$\oplus$ instances are presented as pairs.
For example, $\af{Left} a \in A \oplus B$ becomes $\tup{\af{Left}, a}$.}.

\begin{figure}[t]
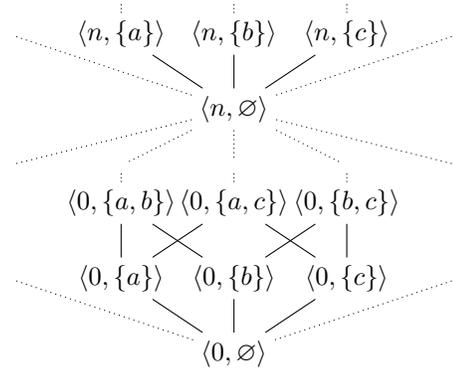

  \centerxy{
    <0.5cm,0pt>:
    (-3,9.5)*+{}="topa";
    (0,9.5)*+{}="topb";
    (3,9.5)*+{}="topc";
    (-6,8.5)*+{}="1l";
    (-3,8.5)*+{\tup{n, \seta}}="1a";
    (0,8.5)*+{\tup{n, \setb}}="1b";
    (3,8.5)*+{\tup{n, \setc}}="1c";
    (6,8.5)*+{}="1r";
    (0,6.5)*+{\tup{n, \varnothing}}="1empty";
    (-6,5)*+{}="0e1";
    (-3,5)*+{}="0e2";
    (0,5)*+{}="0e3";
    (3,5)*+{}="0e4";
    (6,5)*+{}="0e5";
    (-3,5)*+{}="0abn";
    (0,5)*+{}="0acn";
    (3,5)*+{}="0bcn";
    (-3,4)*+{\tup{0, \setab}}="0ab";
    (0,4)*+{\tup{0, \setac}}="0ac";
    (3,4)*+{\tup{0, \setbc}}="0bc";
    (-6,2)*+{}="0l";
    (-3,2)*+{\tup{0, \seta}}="0a";
    (0,2)*+{\tup{0, \setb}}="0b";
    (3,2)*+{\tup{0, \setc}}="0c";
    (6,2)*+{}="0r";
    (0,0)*+{\tup{0, \varnothing}}="0empty";
    "0empty"; "0l" **\dir{.};
    "0empty"; "0a" **\dir{-};
    "0empty"; "0b" **\dir{-};
    "0empty"; "0c" **\dir{-};
    "0empty"; "0r" **\dir{.};
    "0a"; "0ab" **\dir{-};
    "0a"; "0ac" **\dir{-};
    "0b"; "0ab" **\dir{-};
    "0b"; "0bc" **\dir{-};
    "0c"; "0ac" **\dir{-};
    "0c"; "0bc" **\dir{-};
    "0ab"; "0abn" **\dir{.};
    "0ac"; "0acn" **\dir{.};
    "0bc"; "0bcn" **\dir{.};
    "0e1"; "1empty" **\dir{.};
    "0e2"; "1empty" **\dir{.};
    "0e3"; "1empty" **\dir{.};
    "0e4"; "1empty" **\dir{.};
    "0e5"; "1empty" **\dir{.};
    "1empty"; "1l" **\dir{.};
    "1empty"; "1a" **\dir{-};
    "1empty"; "1b" **\dir{-};
    "1empty"; "1c" **\dir{-};
    "1empty"; "1r" **\dir{.};
    "1a"; "topa" **\dir{.};
    "1b"; "topb" **\dir{.};
    "1c"; "topc" **\dir{.};
  }
  \caption{Hasse diagram of $\nat \boxtimes \pow{U}$, for an infinite set $U$,
  where most ideals are infinite.}
  \label{fig:lex-finite}
\end{figure}
  \subsection{Decomposing Compositions}
\label{app:decomposing}

In this section we show that for each composition technique there is a 
corresponding decomposition rule. 
As the lattice join $\join$ of a composite CRDT is defined in terms of
the lattice join of its components \cite{crdt-composition},
decomposition rules of a composite CRDT follow the same idea and resort to
the decomposition of its smaller parts.
We now present such rules for all lattice compositions
covered in Tables \ref{tab:dcc-dist} and \ref{tab:finite}.

\begin{center}
  \vspace{5pt}
  {
\setlength\tabcolsep{.1em}
\begin{tabular}{rll}
  $c$ & $\in C$:&
  \hspace{.3em}$\dec{c} = \{c\}$ \\[3pt]

  $\tup{a, b}$ & $\in A \times B$:&
  \hspace{.3em}$\dec{\tup{a, b}} = \dec{a} \times \{\bot\} \union \{\bot\} \times \dec{b}$ \\[3pt]

  $\tup{c, a}$ & $\in C \boxtimes A$:&
  \hspace{.3em}$\dec{\tup{c, a}} = \dec{c} \times \dec{a}$ \\[3pt]

  $\af{Left} a$ & $\in A \oplus B$:&
  \hspace{.3em}$\dec{\af{Left} a} = \{ \af{Left} v | v \in \dec{a} \}$ \\[3pt]

  $\af{Right} b$ & $\in A \oplus B$:&
  \hspace{.3em}$\dec{\af{Right} b} = \{ \af{Right} v | v \in \dec{b} \}$ \\[3pt]

  $f$ & $\in U {\hookrightarrow} A$:&
  \hspace{.3em}$\dec{f} = \{ k \mapsto v | k \in \dom(f) \land v \in \dec{f(v)} \}$ \\[3pt]

  $s$ & $\in \pow{U}$:&
  \hspace{.3em}$\dec{s} = \{\{e\} | e \in s \}$ \\[3pt]

  $s$ & $\in \mpow{P}$:&
  \hspace{.3em}$\dec{s} = \{\{e\} | e \in s \}$
\end{tabular}
}
  \vspace{4pt}
\end{center}

Note how the decompositions of $\af{GCounter}$ and $\af{GSet}$
presented in Section \ref{subsec:jd} are an application of
these rules.
As a further example, consider a \emph{positive-negative counter} \---
$\af{PNCounter}$ \--- a CRDT counter that allows both increments
and decrements. In this CRDT, each replica identifier is mapped
to a pair where the first component tracks the number of increments,
and the second the number of decrements, i.e.
$\af{PNCounter} = \ids \map (\nat \times \nat)$.
Given a $\af{PNCounter}$ state
$p = \{ \A {\mapsto} \tup{2, 3}, \B {\mapsto} \tup{5, 5} \}$
(2 increments by $\A$, 3 decrements by $\A$, and an equal
number of increments and decrements by $\B$),
the irredundant join decomposition of $p$ is
$\dec{p} = \{
  \{ \A {\mapsto} \tup{2, 0} \},
  \{ \A {\mapsto} \tup{0, 3} \},
  \{ \B {\mapsto} \tup{5, 0} \},
  \{ \B {\mapsto} \tup{0, 5} \}
\}$.

\end{document}